\documentclass[11pt, reqno]{amsart}
\usepackage{fullpage}
\usepackage[T1]{fontenc}
\usepackage{amsthm,amsmath,amsfonts,amssymb,amsaddr}
\usepackage{xcolor}
\usepackage{hyperref}
\hypersetup{
    colorlinks=true,
    linkcolor=purple,    
    urlcolor=purple,
    citecolor=blue
}
\usepackage[pdftex]{graphicx}
\usepackage{booktabs}
\usepackage{array}
\usepackage{multirow}
\usepackage{algorithm}
\usepackage[noend]{algpseudocode}

\newtheorem{theorem}{Theorem}
\newtheorem{corollary}{Corollary}
\newtheorem{lemma}{Lemma}
\newtheorem{proposition}{Proposition}

\theoremstyle{definition}
\newtheorem{definition}{Definition}
\newtheorem{remark}{Remark}
\newtheorem{example}{Example}

\usepackage{mathtools}
\usepackage{caption}
\usepackage{subcaption}

\def\T{{\scriptstyle \top }}

\newcommand{\given}{\mid}

\newcommand{\Real}{\mathbb{R}}
\newcommand{\E}{\mathbb{E}}
\newcommand{\calK}{\mathcal{K}}
\newcommand{\calP}{\mathcal{P}}
\newcommand{\calL}{\mathcal{L}}
\newcommand{\calM}{\mathcal{M}}
\newcommand{\calD}{\mathcal{D}}
\newcommand{\calS}{\mathcal{S}}

\newcommand{\profile}{\mathsf{Profile}}
\newcommand{\Normal}{\mathsf{N}}

\newcommand{\Unif}{\mathsf{Uniform}}

\newcommand\blfootnote[1]{%
  \begingroup
  \renewcommand\thefootnote{}\footnote{#1}%
  \addtocounter{footnote}{-1}%
  \endgroup
}

\usepackage[round]{natbib}

\title{Rashomon-Seeded Annealing for Robust Bayesian Inference in Factorial Designs}

\author{Yiyang Fan$^{1, \ast}$, Soumyakanti Pan$^{1, \ast}$, \and Tyler H. McCormick$^{1, 2}$}
\address{$^{1}$Department of Statistics, University of Washington, Seattle, WA, USA\\ $^{2}$Department of Sociology, University of Washington, Seattle, WA, USA}

\begin{document}
\begin{abstract}
  Integrating over model uncertainty in factorial designs via Bayesian model averaging is hindered by the combinatorial explosion of interpretable interaction effects, often yielding a multimodal posterior, where standard Markov chain Monte Carlo algorithms encounter significant convergence issues. We propose a general computational framework that repurposes Rashomon sets, collections of high‑performing models traditionally valued for prediction and interpretability, as a strategic ``warm start'' for estimating the full posterior. Our method, Rashomon‑seeded annealing, initializes annealed importance sampling (AIS) by anchoring the starting density within these pre-identified, high-evidence regions while preserving global support over the entire model space. Rather than restricting inference to the Rashomon set and understating uncertainty, the AIS correction restores full posterior inference, turning the Rashomon certificate from an inferential truncation into a proposal mechanism. We demonstrate this approach using Rashomon Partition Sets (RPS) as a rigorous, certified seed constructor for factorial designs. The resulting algorithm yields consistent self-normalized posterior summaries, such as model‑averaged cell means, credible intervals, and uncertainty summaries without exhaustive enumeration of the complete model space. This bridges the gap between high-evidence model discovery and rigorous Bayesian inference, and outlines a general strategy in which any high‑posterior seed set can provide computational leverage for AIS‑based model averaging.

\end{abstract}
\maketitle
\blfootnote{$^{\ast}$Co-first authors.}
\vspace{-28pt}

\section{Introduction}\label{sec:intro}

Bayesian model averaging (BMA) provides a coherent probabilistic framework for inference when the true model is unknown \citep{Madigan1994, madigan1996bma, Raftery1997_bmalm, hoeting1999bma}. Let $\calM = \{M_1, \dots, M_G\}$ denote a finite collection of candidate models for an observed data $\calD$, where each $M_g$ is associated with a probability model $p(\calD \given M_g)$. BMA obtains the posterior distribution of any quantity of interest by marginalizing over model uncertainty as
\begin{equation}\label{eq:bma}
    p(\cdot \given \calD) = \sum_{g=1}^{G} p(\cdot \given \calD, M_g) \; p(M_g \given \calD) \;,
\end{equation}
where model-specific inferences are weighted by their respective posterior model probabilities, $p(M_g \given \calD) \propto p(\calD \given M_g) \, p(M_g)$. For example, in factorial designs, BMA addresses the uncertainty inherent in parsimonious representation. While treatments comprise multiple features at varying levels, outcome variation is often driven by only a subset of interactions, so that identifying which treatment combinations to ``pool'' for a simplified interaction structure generates a vast space of candidate hypotheses. Thus, while BMA provides a principled quantification of uncertainty, it is often computationally prohibitive when $|\calM|$ is large --- for instance, growing combinatorially with the number of model components/parameters. In such cases, exhaustive enumeration of $\calM$ is deemed impractical. Moreover, when the posterior over $\calM$ is multimodal, traditional Markov chain Monte Carlo (MCMC) algorithms frequently struggle to navigate the model space, resulting in poor mixing and significant convergence issues \citep[see for e.g.,][]{Volinsky1997, Guan2011, Yang2016}. These challenges have motivated various strategies that approximate BMA by concentrating computational resources on high-posterior regions of the model space. 

One influential approach is Occam’s window \citep{Madigan1994}, which restricts model averaging to a subset not decisively inferior to the MAP model. Similar principles underpin stochastic search variable selection \citep{George1993}, mode-jumping MCMC \citep{Hubin2018_EBMS}, and shotgun stochastic search \citep{Hans2007}, model scan \citep{Chib2020_ModelScan}, which identify high-probability models to improve exploration. This selective paradigm also aligns with PAC-Bayesian theory \citep{McAllester1999, Guedj2019}, which provides generalization guarantees by concentrating density on low-risk regions of the model space. Collectively, these strategies aim to make inference tractable by targeting the finite collection of models capturing the high-posterior regions. Rashomon sets formalize this paradigm by aggregating all models within a specified tolerance of the \emph{maximum a posteriori} (MAP) model \citep{breiman_twocultures}. In our context, we define the Rashomon set $\mathcal{R}_\epsilon$ as the collection of all models whose posterior evidence is within a tolerance $\epsilon \in [0, 1]$ of the maximum a posteriori model $M^\ast \in \calM$, assuming it exists, as
\begin{equation}\label{eq:rashomon_set}
\mathcal{R}_\epsilon = \left\{ M \in \calM : \log \tilde{p}(M \given \calD) \geq \log \tilde{p}(M^\ast \given \calD) - \epsilon \right\}\;,
\end{equation} 
where $\tilde{p}(M \given \calD)$ denotes to the unnormalized posterior density corresponding to $p(M \given \calD)$.

However, conditioning inference on a high-posterior subset of $\calM$, such as the Rashomon set $\mathcal{R}_\epsilon$, entails an inferential truncation of \eqref{eq:bma}. While previous literature utilizes these sets to ensure predictive robustness or to provide interpretable model summaries \citep{Rudin2019_rashomon, Dong2020, Marx2020, xin2022, semenova2022}, relying on them for posterior approximation necessarily neglects the mass beyond $\mathcal{R}_\epsilon$ and understates model uncertainty. In this paper, we propose a distinct utility for these sets. Rather than approximating posterior inference through truncation, we frame the Rashomon set as a strategic ``seed'' to accelerate the computation of the full posterior. The central idea is that a pre-identified collection of high-evidence models can be embedded into a proposal distribution to initialize annealed importance sampling \citep[AIS,][]{gelman_xlmeng_1998, neal2001_ais} while maintaining full support over $\mathcal{M}$. Although AIS and its related variants has primarily been used to estimate marginal likelihoods \citep{Lartillot2006, Friel2008, Baele2012, Fan2011} for model comparison, its potential as a mechanism for sampling over a discrete model space for Bayesian model averaging, when initialized with the Rashomon set, remains, to the best of our knowledge, unexplored. Our proposed framework, termed \emph{Rashomon-seeded annealing}, provides a ``warm start'' by anchoring the initialization of AIS in the high-evidence models within the Rashomon set, and subsequently overcome BMA convergence bottlenecks without the inferential cost of truncation. The subsequent AIS correction restores globally consistent inference without requiring exhaustive enumeration of the model space.

We illustrate this general framework using factorial designs, where $\calM$ comprises permissible partitions of the feature space that group indistinguishable feature-level combinations \citep{apara2025, banerjeeTVA}. In these settings, the Rashomon Partition Set (RPS) provides a deterministic, geometry-aware enumeration of all partitions within a specified posterior threshold of the MAP model. As a certified seed set, the RPS is uniquely suited to initialize seeded annealing, allowing it to circumvent the mixing issues of multimodal partition spaces while retaining the rigor of full-posterior BMA. 
The remainder of this manuscript is organized as follows. Section~\ref{sec:annealing} develops the general Rashomon-seeded annealing framework. Section~\ref{sec:setup} details its application to Bayesian factorial designs. Section~\ref{sec:simulation} evaluates the proposed approach across various simulation settings, and Section~\ref{sec:data_analysis} demonstrates its utility using the NHANES telomere length, and charitable donation dataset. Finally, Section~\ref{sec:discussion} offers concluding remarks and suggests directions for future research. All code and datasets utilized in this work are available at \url{https://anonymous.4open.science/r/RashomonSeededAnnealing-DE5F}.

\section{Rashomon-seeded annealing}\label{sec:annealing}

The framework we develop here is targeted towards model spaces that possess a
structural geometry, typically the setting where the size of the model space $|\calM|$ grows combinatorially. In such spaces, a meaningful distance metric $d: \calM \times \calM \to \mathbb{Z}_{\geq 0}$ typically exists to quantify the number of elementary structural changes separating any two models, where $\mathbb{Z}_{\geq 0}$ corresponds to non-negative integers. The specific definition of $d$ remains context-dependent. For example, in variable selection, $d$ could be the Hamming distance between inclusion vectors; in spaces of trees, $d$ could be the Robinson–Foulds distance or the minimum number of subtree prune‑and‑regraft operations; in factorial designs (see Section~\ref{sec:setup}), models are permissible partitions of the feature space, and $d$ is the minimum number of cut‑swaps needed to transform one partition into another. Thus, given a $d$, we define the following. A metric $d$ naturally defines a $\delta$-neighborhood of any model $M$ by $\mathcal{N}_\delta(M) = \{M'\in\calM: d(M,M')\leq \delta\}$, for some $\delta \in \mathbb{Z}_{\geq 0}$.
\begin{definition}[Level sets]
  For any model $M \in \calM$ and a collection of models $\calS \subset \calM$, define $d(M, \calS) = \min_{M' \in \calS} \; d(M, M')$. Then, for $k = 0, 1, 2, \ldots$, the $k$‑th level set of $\calS$ is defined as
  \[
  \calL_k(\calS) = \left\{ M \in \calM : d(M, \calS) = k \right\}\;.
  \]
\end{definition}

\subsection{Annealed importance sampling}\label{sec:ais}
\subsubsection{Initial distribution from a seed set.}
Let $\calS \subset \calM$ denote a pre-identified non-empty collection of models, defined as a seed set, and assume that the unnormalized posterior $\tilde{p}(M \given \calD)$ is evaluable for any $M \in \calM$. To ensure global exploration while concentrating mass on high-evidence regions, we construct the unnormalized initial proposal $q_0(M; \calS)$ as a mixture over three components: the seed set $\calS$, its immediate level set $\mathcal{L}_1(\calS)$, and a uniform component over the entire model space $\calM$. For some $\alpha_1, \alpha_2 > 0$ with $\alpha_1 + \alpha_2 < 1$, define the unnormalized initial density
\begin{equation}\label{eq:q0}
  q_0(M ; \calS) \;=\;
  \alpha_1 \, \frac{\tilde{p}(M \given \calD)}{\sum_{M' \in \calS} \tilde{p}(M' \given \calD)}\, \mathbf{1}_{\calS}(M) + \alpha_2^{+} \, \frac{\mathbf{1}_{\mathcal{L}_1(\calS)}(M)}{|\mathcal{L}_1(\calS)|} + (1 - \alpha_1 - \alpha_2^{+}) \;,
\end{equation}
where $\tilde{p}(M \given \calD)$ denote the unnormalized posterior probabilities of model $\calM \in \calS$, and $\mathbf{1}_{A}(\cdot)$ denotes the indicator function over $A \subseteq \calM$, and $\alpha_2^{+} = \alpha_2\, \mathbb{I}(\calL_1(\calS) \neq \emptyset)$, thereby redistributing the mass to the global uniform component if the unit-neighborhood of the seed set is empty.. The third term is simply corresponds to an unnormalized uniform density over $\calM$. This construction guarantees $q_0(M; \calS) > 0$ for every $M \in \calM$, i.e., the proposal has full support over $\calM$. Importantly, evaluating $q_0(M; \calS)$ does not require enumerating $\calM$ or computing the normalizing constant of the uniform distribution; the unknown constant cancels in the self-normalized importance weights of the AIS algorithm, which we discuss later in the following sections, leaving the resulting posterior estimates unaffected.
Sampling from $q_0$ is straightforward: draw a uniform random variable $U \sim \Unif (0,1)$; if $U < \alpha_1$, sample $M$ from $\calS$ with probability proportional to $\tilde{p}(M \given \calD)$; if $\alpha_1 \leq U < \alpha_1 + \alpha_2$, sample uniformly from $\calL_1(\calS)$; otherwise, $M$ is sampled uniformly from the complete model space $\calM$, typically via a convenient representation that maintains a bijection to $\mathcal{M}$, such as the $\Sigma$-matrix for factorial partitions \citep{apara2025}.

\subsubsection{Algorithm}
Given the unnormalized initial density $q_0(M; \calS)$ defined in \eqref{eq:q0} and the unnormalized posterior $\tilde{p}(M \given \calD)$, we construct a sequence of intermediate distributions that bridge the initial density to the target posterior. Fix a temperature ladder $0 = b_0 < b_1 < \dots < b_T = 1$ and define the unnormalized intermediate densities for $t = 0, \ldots, T$ as
\begin{equation}\label{eq:bridge}
  f_t(M) = q_0(M; \calS)^{1-b_t}\, \tilde{p}(M \given \calD)^{b_t}, \quad M \in \calM \;.
\end{equation}
As $b_t$ increases, the influence of the posterior density grows while that of the initial decays, maintaining full support over $\calM$ at every intermediate step. Crucially, because the intermediate densities are unnormalized, their respective normalizing constants cancel within the importance weights. Thus, the algorithm requires only the pointwise evaluation of the unnormalized densities $f_t$.

To generate a single AIS trajectory, an initial model $M^{(0)}$ is sampled from $q_0(\cdot; \calS)$ and assigned a unit weight. At each temperature step $b_t$, the model state is refreshed via a Metropolis–Hastings transition kernel that targets the intermediate distribution $f_t$. This update employs a symmetric proposal $\phi(\cdot \mid M^{(t-1)})$ supported on the unit-neighborhood $\mathcal{N}_1(M^{(t-1)})$, a step that leaves the target $f_t$ invariant \citep{tokdar2010}. The importance weight for $M^{(t)}$ is subsequently updated recursively by the ratio $f_t(M^{(t-1)}) / f_{t-1}(M^{(t-1)})$. Repeating this process $J$ times yields a collection of properly weighted samples $\{(M_j, w_j)\}_{j=1}^J$, where $M_j$ denotes the terminal state of the $j$-th trajectory and $w_j$ the importance weight associated with $M_j$. See Algorithm~\ref{algo:seeded_ais} for details.
\begin{algorithm}[b]
  \caption{Seeded Annealed Importance Sampling}
  \label{algo:seeded_ais}
  \begin{algorithmic}[1]
    \State \textbf{Input:} seed set \(\calS\), level set \(\mathcal{L}_1(\calS)\),
           unnormalized posterior \(u(M)\), temperatures \(\{b_t\}_{t=1}^{T}\),
           number of trajectories \(J\), number of Metropolis-Hatings updates \(L\).
    \State Construct unnormalized initial proposal \(q_0(\cdot; \calS)\) via \eqref{eq:q0}.
    \For{\(j = 1,\ldots,J\)}
      \State Draw \(M_j^{(0)} \sim q_0\) (self‑normalized); set \(w_j = 1\).
      \For{\(t = 1,\ldots,T\)}
        \State Starting from \(M_j^{(t-1)}\), run \(L\) Metropolis–Hastings steps
               with symmetric proposal \(\phi(\cdot \given M_j^{(t-1)})\) targeting \(f_t\).
        \State Let \(M_j^{(t)}\) be the final state.
        \State \(w_j \leftarrow w_j \cdot f_t(M_j^{(t-1)}) \,/\, f_{t-1}(M_j^{(t-1)})\).
      \EndFor
      \State Store \((M_j, w_j)\) with \(M_j = M_j^{(T)}\).
    \EndFor
    \State \textbf{Output:} weighted sample \(\{(M_j, w_j)\}_{j=1}^J\)
           targeting the posterior proportional to \(u(M)\).
  \end{algorithmic}
\end{algorithm}

\subsection{Theoretical guarantees}\label{sec:theory}
Let $Q$ be the joint distribution of a model-weight pair $(M, w)$. By construction, independent chains yield i.i.d. replicates $(M_j, w_j) \sim Q$ for $j = 1, \ldots, J$, and we let $\mathbb{Q}$ denote the product probability measure on $(\calM \times (0, \infty))^{\mathbb{N}}$ induced by $Q$. The importance weights satisfy the fundamental unbiasedness property of annealed importance sampling \citep{gelman_xlmeng_1998, neal2001_ais}. That is, if $C_T = \sum_{M \in \calM} \tilde{p}(\cdot \given \mathcal{D})$ and $C_0 = \sum_{M \in \calM} q_0(M; \calS)$ are the normalizing constants of the unnormalized posterior $\tilde{p}(\cdot \given \calD)$ and the unnormalized initial proposal $q_0(\cdot; \calS)$, respectively, then for any bounded measurable function $\zeta : \calM \to \Real^p$, there exists $C = C_T / C_0 > 0$ such that $\E_{Q} \left[ w \, \zeta(M) \right] = C \, \E_{M \given \calD} \left[ \zeta(M) \right]$, where the expectations are taken component‑wise. In particular, $\E_{Q}[w] = C$. We assume that
$\E_{M \given \calD}[ \lVert \zeta(M) \rVert ] < \infty$ and define the self‑normalized estimator
\begin{equation}\label{eq:ais_estimator_general}
  \hat{\psi}_J = \frac{\sum_{j=1}^J w_j \, \zeta (M_j)}{\sum_{j=1}^J w_j} \;.
\end{equation}
\begin{theorem}[Consistency of AIS estimators]\label{thm:general}
  Assume $\E_{Q}[w] < \infty$ and $\E_{Q}[ w \, \lVert \zeta(M) \rVert ] < \infty$. Then $\hat{\psi}_J \;\xrightarrow{\text{a.s.}} \psi$, where $\psi = \mathbb{E}_{M \given \calD} [\zeta(M)]$, as $J \to \infty$, with almost‑sure convergence in $\mathbb{R}^p$ under $\mathbb{Q}$.
\end{theorem}

\begin{corollary}[Posterior distribution function]\label{cor:posterior_func}
  Let $\theta$ be a scalar parameter of interest and denote the conditional posterior distribution function of $\theta$ under model $M$ by $P(\theta \given \calD, M)$. Then
  \[
    \hat{P}_J(\theta) \;\xrightarrow{\;\text{a.s.}\;} \;
    P(\theta \given \calD) = \sum_{M \in \calM} P(\theta \given \calD, M) \, p(M \given \calD) \;,
  \]
  where $\hat{P}_J(\theta) = \frac{\sum_{j=1}^J w_j \, P(\theta \given \calD, M_j)}{\sum_{j=1}^J w_j}$, which is obtained by plugging in $\zeta(M) = p(\theta \given \calD, M)$ in \eqref{eq:ais_estimator_general}.
\end{corollary}

\begin{corollary}[Uniform convergence and quantile consistency]\label{cor:uniform}
  Under the assumptions of Theorem~\ref{thm:general},
  \[
    \sup_{\theta \in \mathbb{R}} \bigl| \hat{P}_J(\theta) - P(\theta \given \calD) \bigr|
    \xrightarrow{\text{a.s.}} 0,
  \]
  as $J \to \infty$.  If, in addition, $P(\theta \given \calD)$ is continuous and
  strictly increasing in a neighbourhood of its $\alpha$-quantile $q_\alpha$, then the
  empirical quantile $\hat{q}_{\alpha,J} = \inf\{\theta : \hat{P}_J(\theta) \ge \alpha\}$
  satisfies $\hat{q}_{\alpha,J} \xrightarrow{\text{a.s.}} q_\alpha$.
\end{corollary}
The proofs of the results are in Appendix~\ref{sec:consistency}. The unbiasedness of the importance weights is a standard property of annealed importance sampling \citep{neal2001_ais, gelman_xlmeng_1998}, and the consistency argument follows from the strong law of large numbers applied to the self‑normalized estimator; the proofs rely only on the full support of $q_0$ and the validity of the intermediate Metropolis–Hastings kernels, not on any special property of the seed set $\mathcal{S}$. 

In practice, however, selecting a seed set from high-posterior regions, such as the Rashomon set, drastically improves efficiency; an initial distribution near the target requires fewer temperatures ($T$) and updates ($L$) to maintain a stable effective sample size. In practice, we utilize the self-normalized weights $\bar{w}_j = w_j / \sum_{j'} w_{j'}$ to approximate posterior moments, marginal densities, and credible intervals via the estimator $\hat{P}_J(\theta) = \sum \bar{w}_j P(\theta \mid \mathcal{D}, M_j)$. To minimize weight variance \citep{Gneiting2007}, we recommend a log-equispaced temperature ladder \citep{CALDERHEAD2009, Zhou2016} or an adaptive pilot run to preserve effective sample size throughout the bridge \citep[Appendix~\ref{sec:adaptive_ladder};][]{cameron_2019}.

\section{Bayesian factorial designs and the Rashomon Partition Set}\label{sec:setup}

Let $y = (y_1, \dots, y_n)^\T$ denote the vector of $n$ observations, and let $X$ be an $n \times M$ matrix whose $i$-th row $x_i = (x_{i1}, \dots, x_{iM})^\T$ records the levels of $M$ features for observation $i$. Here $\calD \coloneqq (y, X)$. In general, each feature $m \in \{1, \ldots, M\}$ is partially ordered and takes one of $R_m$ levels, resulting in $K = \prod_{m=1}^M R_m$ distinct factor-level combinations, or cells, which comprise the set $\mathcal{K}$. However, for notational clarity and without loss of generality, we assume a constant number of levels, $R$, for all features, as the underlying methodology remains analogous. 

To discover interpretable interaction effects within this setting, we consider the \emph{treatment variant aggregation} framework \citep{banerjeeTVA}, where the central task is to partition $\calK$ by merging cells whose expected outcomes are statistically indistinguishable. Such aggregation must respect the partial order, i.e., only cells that are comparable under increasing levels can be grouped, ensuring that aggregated units remain scientifically meaningful (see Appendix~\ref{sec:geometry_app} for details). We define a \emph{pool} as a collection of cells that are constrained to share the same cell means, where a cell mean is the expected outcome under the feature combination of the cell. A \emph{partition} $\Pi$ is a collection of disjoint pools that cover $\calK$. Partitions must be \emph{permissible}, that is, their pools must be contiguous in the partial order and satisfy parallel‑split constraints (see
Figure~\ref{fig:hasse_main}).
\begin{figure}[t]
  \centering
  % First Subfigure
  \begin{subfigure}[b]{0.32\textwidth}
    \centering
    \includegraphics[width=\textwidth]{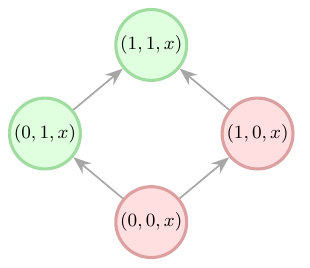}
    \caption{Permissible}
    \label{fig:hasse1}
  \end{subfigure}
  \hfill % Adds spacing between the images
  % Second Subfigure
  \begin{subfigure}[b]{0.32\textwidth}
    \centering
    \includegraphics[width=\textwidth]{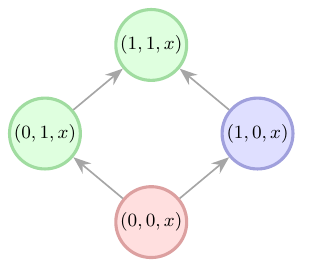}
    \caption{Not permissible}
    \label{fig:hasse2}
  \end{subfigure}
  \hfill
  % Third Subfigure
  \begin{subfigure}[b]{0.32\textwidth}
    \centering
    \includegraphics[width=\textwidth]{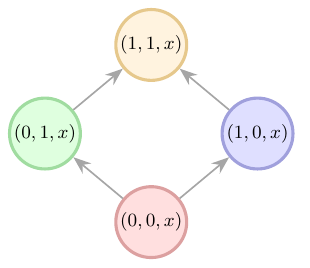}
    \caption{Permissible}
    \label{fig:hasse3}
  \end{subfigure}
  \caption{Hasse diagrams for Example~\ref{ex:factorial_defs} with the third feature fixed at level $x$. (a) A permissible partition with two pools; (b) a non-permissible partition with three pools; and (c) the saturated case where each cell forms its own pool. Distinct colors distinguish separate pools.}
  \label{fig:hasse_main}
\end{figure}
We detail the permissibility rules in Appendix~\ref{sec:geometry_app}. Here, all permissible partitions $\calP^\ast$ instantiates the model space $\mathcal{M}$.
\begin{example}[Pools and permissibility]\label{ex:factorial_defs}
  Figure~\ref{fig:hasse_main} shows a design with $M=3$ features, each taking
  levels $0$ or $1$, with $0\prec 1$, where $\prec$ denotes the ``order''. The Hasse diagrams depict the eight possible cells, but for visual clarity we focus on the slice where the third feature is fixed at level $x$, leaving four displayed cells. Figure~\ref{fig:hasse1} shows a permissible partition with two pools: $\{(0,0,x), (1,0,x)\}$ (blue) and $\{(0,1,x), (1,1,x)\}$ (red). This pooling states that the first feature has no effect on the expected outcome, regardless of the second feature's level. Figure~\ref{fig:hasse2} is non‑permissible, with the singleton pool $(1,0,x)$ (red) and the two‑cell pool $\{(0,1,x),(1,1,x)\}$ (blue) violating the parallel‑split rule. Figure~\ref{fig:hasse3} is the saturated permissible partition where every cell is its own pool.
\end{example}

\subsection{Bayesian hierarchical model}\label{sec:bayes_model}
Given a permissible partition $\Pi = \{\pi_1, \dots, \pi_H\}$, the cell means $\beta \in \Real^K$ is constrained such that for each $h = 1, \ldots, H$, if feature combinations $k^{(1)}, k^{(2)} \in \pi_h$, then $\beta_{k^{(1)}} = \beta_{k^{(2)}}$. Let $\gamma \in \Real^H$ denote pool means, where $\gamma_h$ represents the common expected value for all cells in pool $\pi_h$. Then we bridge the cell means $\beta$ and the pool means $\gamma$ via a binary transformation matrix $\Lambda_{\Pi} \in \{0, 1\}^{K \times H}$, as $(\Lambda_{\Pi})_{kh} = 1$ if $k \in \pi_h$, and 0 otherwise.
Since $\Pi$ is a valid partition, each row of $\Lambda_{\Pi}$ contains exactly one entry of $1$, ensuring that every feature combination maps to a unique pool. Hence, we can write the cell means as a linear transformation of the pool means, $\beta = \Lambda_{\Pi} \gamma$. Let $\tilde{X}$ be the $n\times K$ binary indicator matrix with $\tilde{X}_{ik}=1$ iff observation $i$ belongs to cell $k$. Moreover, we write $\tilde{X}_\Pi \coloneqq \tilde{X} \Lambda_\Pi$. Subsequently, we define the Bayesian hierarchical factorial model as
\begin{equation}\label{eq:hier}
\begin{split}
    y \given \gamma, \Pi, \sigma^2 &\sim \Normal (\tilde{X}_\Pi \gamma, \sigma^2 I_n)\;,\\
    \gamma \given \Pi, \sigma^2 &\sim \Normal (0, g\sigma^2 (\tilde{X}_\Pi^\T \tilde{X}_\Pi)^{-1}),\quad 
    p(\Pi) \propto \exp \{- \lambda |\Pi|\}\;,
\end{split}
\end{equation}
where $\lambda >0$, and $\lvert \Pi \rvert$ denotes the number of distinct pools in $\Pi$, i.e., the size of the partition $\Pi$. For pool means $\gamma$, we assign a $g$-prior \citep{ZellnerSiow1980, Zellner1986}. By construction, we have $\tilde{X}_\Pi^\T \tilde{X}_\Pi = \mathrm{diag} (n_1, \ldots, n_H)$, where $n_h > 0$ denotes the number of observations in pool $\pi_h \in \Pi$ for each $h$ (see Proposition~\ref{prop1} in Appendix~\ref{sec:appendix_posterior}). Hence, the $g$-prior introduces a shrinkage that stabilizes the estimates of the pool means, adjusting for the number of observations in each pool, while preserving conjugacy. Additionally, the prior on $\Pi$ plays a regularizing role by favoring less granular aggregations. It essentially corresponds to an $\ell_0$ penalty, similar in spirit to the Occam's window \citep{Madigan1994, madigan1996bma}, where conditional on the number of pools in a partition, all permissible partitions in $\calP^\ast$ are equally likely. As we focus on $\calP^\ast$, we treat $\sigma^2$ as a nuisance parameter. We estimate $\sigma^2$ separately from \eqref{eq:hier}, e.g., estimated from the saturated model (see Figure~\ref{fig:hasse3}) and thereafter assumed to be fixed in the hierarchical model in \eqref{eq:hier}.

\subsubsection{Posterior distribution} 
Assume $\sigma^2$ is known. Then the joint posterior distribution for \eqref{eq:hier} is written as $p(\beta, \Pi \given \calD) = p(\beta \given \Pi, \calD) \times p(\Pi \given \calD)$. For $k = 1, \ldots, K$, given a partition $\Pi = \{ \pi_1, \ldots, \pi_H \}$, the theory of conjugate priors yields
\begin{equation}\label{eq:post_beta}
    \beta_k \given \Pi, \calD \sim \Normal \left( \frac{g}{g+1} \hat\gamma_{\Pi, h}, \frac{g}{g+1} \cdot \frac{\sigma^2}{n_h} \right), \quad \text{if } k \in \pi_h  \;,
\end{equation}
where $\hat{\gamma}_{\Pi, h} = \frac{1}{n_h} \sum_{i:k(i) \in \pi_h} y_i$, with $k(i) \in \calK$ denoting the feature combination corresponding to the observation $y_i$. See Appendix~\ref{sec:gamma_app} for details. Unlike the conditional posterior of $\beta$, the model posterior $p(\Pi \given \calD)$, does not admit a standard distributional form. We derive a closed form for $p(\Pi \given \calD)$ by integrating out $\gamma$ from \eqref{eq:hier}, assuming $\sigma^2$ fixed. Subsequently, we show that
\begin{equation}\label{eq:QPi}
    \arg \max_{\Pi \in \calP^\ast} p(\Pi \given \calD) = \arg \min_{\Pi \in \calP^\ast} \left( \tfrac{1}{n} \mathrm{SSE}_\Pi + \tilde{\lambda} |\Pi| \right) \;,
\end{equation}
where $\tilde{\lambda} > 0$, and $\mathrm{SSE}_{\Pi} = \sum_{h = 1}^H \sum_{i:k(i) \in \pi_h} (y_i - \hat{\gamma}_h)^2$ (see Appendix~\ref{sec:equivalence_app} for details). The term linear in $|\Pi|$ essentially incorporates the $\ell_0$ penalty from the prior. %Hence, the maximum a posteriori partition $\Pi^{\text{MAP}}$ is the minimizer of $Q(\Pi)$, and the posterior is recovered exactly as $p(\Pi \given Z) \propto \exp\{-Q(\Pi)\}$. Given the marginal posterior probability of each permissible partition, posterior inference of the cell means proceed via \eqref{eq:bma}.

\subsection{Rashomon-seeded annealing for factorial designs}\label{sec:rps}
For the Bayesian factorial model in \eqref{eq:hier}, \citet{apara2025} develop an algorithm to efficiently enumerate the Rashomon set $\mathcal{R}_\epsilon$, as described in \eqref{eq:rashomon_set}, within the space of permissible partitions $\calP^\ast$, and termed as Rashomon Partition Set (RPS). While posterior inference is subsequently restricted to $\mathcal{R}_\epsilon$, such a truncation risks underestimating model uncertainty. We address this limitation by utilizing the RPS exclusively as a seed set, i.e., $\calS = \mathcal{R}_\epsilon$, for the seeded annealing framework introduced in Section~\ref{sec:annealing}.

To instantiate seeded annealing in this setup, we equip $\calP^\ast$ with a suitable distance metric $d$. Following \citet{apara2025}, any permissible partition is uniquely identified by a canonical cut-incidence matrix $\Sigma$, where entries indicate the positions of cuts along the Hasse diagram edge families (see Figure~\ref{fig:hasse_main}). Here, the distance $d(\Pi_1, \Pi_2)$ is defined as the edit distance between the corresponding $\Sigma$-matrices. Under this metric, the unit-neighborhood of a partition consists of all permissible partitions reachable by flipping a single entry of its $\Sigma$-matrix. Consequently, the level set $\calL_1(\mathcal{R}_\epsilon)$ in the initial distribution \eqref{eq:q0} comprises all permissible partitions exactly one ``cut-swap'' away from some member of the RPS, $\mathcal{R}_\epsilon$. Thus, with the seed set $\calS=\mathcal{R}_\epsilon$ and the metric $d$ defined above, the initial distribution $p_0(\cdot;\mathcal{R}_\epsilon)$ follows directly from \eqref{eq:q0}, and the seeded annealing algorithm, Algorithm~\ref{algo:seeded_ais} applies without modification. Upon sampling $J$ independent AIS trajectories, we obtain partition-weight pairs $\{(\Pi_j, w_j)\}_{j=1}^J$. We subsequently estimate the marginal posterior of the cell means, $p(\beta \given \calD)$, by applying the AIS estimator \eqref{eq:ais_estimator_general} with the choice $\zeta(\Pi) = p(\beta \given \Pi, \calD)$. Given that the conditional posterior $p(\beta \given \Pi, \calD)$ is Gaussian per \eqref{eq:post_beta}, the resulting AIS estimate for $p(\beta \given \calD)$ is essentially as a Gaussian mixture weighted by the importance weights $\{w_j\}_{j=1}^J$. This analytical tractability facilitates straightforward derivation of downstream posterior summaries, such as the mean and credible intervals, while accounting for the full scope of model uncertainty across $\mathcal{P}^\ast$.

\section{Simulation}\label{sec:simulation}

We evaluate the robustness of our framework by investigating how the initial seed set affects the final AIS inference. Our goal is to see how well our proposed algorithm estimates relevant summary statistics of the full posterior distribution given a RPS of varying size as the seed set, and to compare its performance with other methods. This attempts to show whether AIS provides a reliable correction across different settings or if it remains sensitive to the boundaries of the starting Rashomon set. Besides RPS-truncated inference and the annealing, we also evaluate a PAC-Bayesian (PB) estimator, which serves as a non-annealed competitor utilizing high-density regions for inference. We evaluate the framework across two simulated scenarios of increasing complexity. \emph{Scenario~1} features a modest model space with $M=2$ features at $R_1=4$ and $R_2= 3$ levels ($K=12$ cells), allowing for the computation of the exact posterior via exhaustive enumeration as a definitive baseline. \emph{Scenario~2} considers the setting with $M=3$ features at $R_1 = 4$, $R_2 = 3$ and $R_3 = 3$ levels ($K=36$ cells); here, the combinatorial growth of $\calM$ necessitates the use of an extensive MCMC run (approximately $10^6$ iterations) to provide a ``ground truth'' surrogate. 

Across 50 independent replications, we evaluate the precision of the resulting posterior summaries, specifically the mean and 95\% credible intervals, using a mixture quantile solver (see Appendix~\ref{sec:ais_app}). All estimators ultimately represent the posterior of $\beta$ as a weighted mixture of Gaussian densities (see Section~\ref{sec:setup}); what sets them apart, however, is the strategy they use to explore the model space $\calM$. For the MCMC benchmark, we employ a Metropolis–Hastings algorithm that refreshses the model state by flipping a single, randomly selected entry of the canonical cut-incidence matrix $\Sigma$. Because $\Sigma$ maintains a bijection with $\mathcal{P}^\ast$, every such proposal is guaranteed to be permissible, circumventing the need for complex feasibility checks during the random walk.

\subsubsection{Inferential Accuracy and Uncertainty Quantification}
We evaluate the precision of the estimated posterior summaries relative to the reference posterior --- the exact distribution for Scenario~1 and the MCMC posterior for Scenario~2. For each summary statistic (posterior mean, $2.5\%$, and $97.5\%$ quantiles), we compute the $L_1$ error as the mean of absolute deviations from the corresponding reference value across all cells. 
Notably, we observe that the accuracy of Rashomon-seeded annealing (henceforth AIS) remains remarkably stable as the Rashomon threshold $\epsilon$ increases (see Figure~\ref{fig:accuracy_scenario2}). 
\begin{figure}[t]
    \centering
    \includegraphics[width=0.95\linewidth]{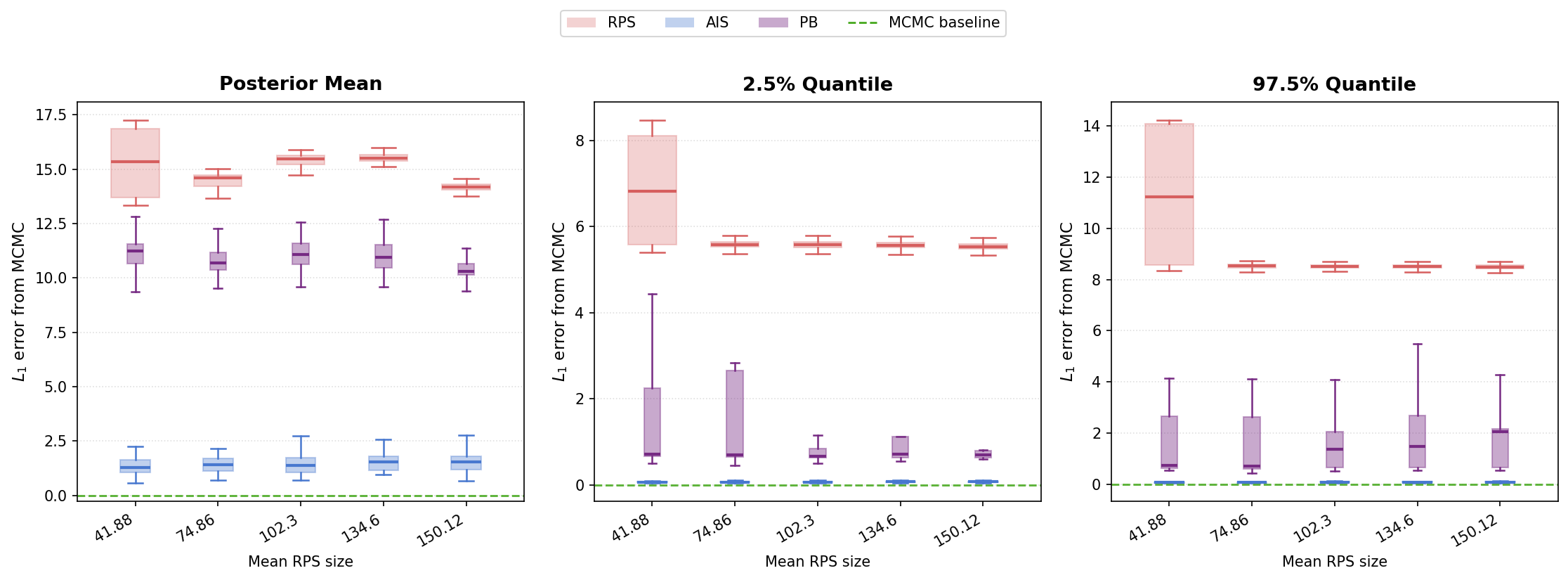}
    \caption{Inferential accuracy relative to the exact posterior: $L_1$ deviation of posterior summaries across varying Rashomon thresholds $\epsilon$. Rashomon-seeded annealing (AIS) consistently outperforms RPS-truncation and PAC-Bayesian (PB), approaching the MCMC posterior as the seed set expands.}
    \label{fig:accuracy_scenario2}
\end{figure}
This indicates that the annealed correction is robust to the initial seed volume; even a small RPS provides a sufficient warm start for the AIS to navigate the model space effectively. Consequently, one can achieve considerably good uncertainty quantification without the computational burden of enumerating an expansive Rashomon set for a seed.
To further assess the fidelity of uncertainty quantification, we compute the Intersection-over-Union (IoU) metric between the estimated and reference credible intervals. To be specific, for any two intervals $I_1$ and $I_2$, $\mathrm{IoU}(I_1, I_2) = {|I_1 \cap I_2|}/{|I_1 \cup I_2|}$, where $|\cdot|$ denotes the Lebesgue measure (the total length) of the set. An IoU of 1 indicates perfect alignment with the reference interval, whereas lower values signify distortions from the reference interval. In Scenario~2, Rashomon-seeded annealing (henceforth AIS) demonstrates a superior ability to recover the global posterior mass compared to the localized alternatives(see Figure~\ref{fig:IoUtime_scenario2}). 
\begin{figure}[t]
    \centering
    \includegraphics[width=0.75\linewidth]{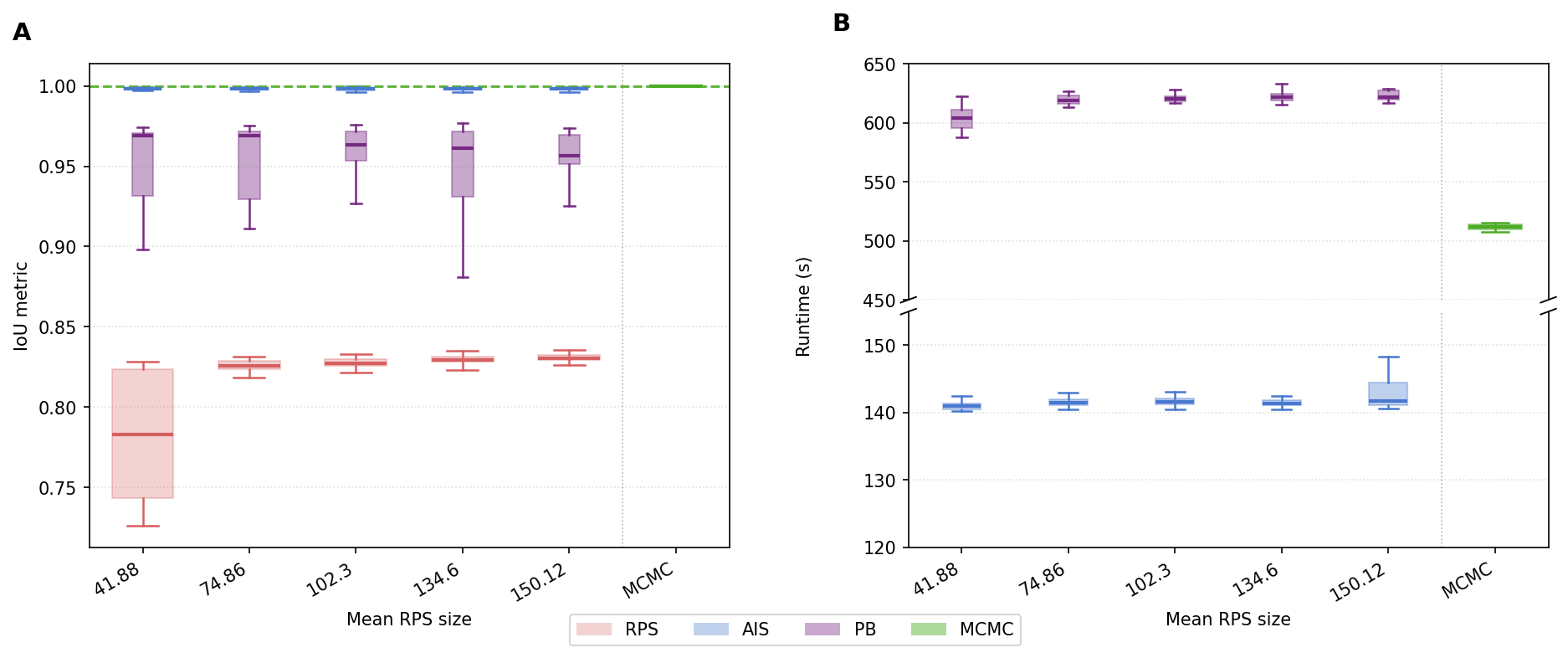}
    \caption{\textbf{A}:Comparison between the loU metric for agreement of intervals produced by MCMC, for RPS, AIS, PB with different $\epsilon$; \textbf{B}:running time for AIS and PB for different $\epsilon$, and MCMC}
    \label{fig:IoUtime_scenario2}
\end{figure}
While both AIS and PB improve upon the RPS-truncated inference, AIS achieves substantial reductions in $L_1$ error across all summaries, including the tail quantiles. In contrast, PB provides only modest gains, primarily centered on the posterior mean. This discrepancy is most pronounced in the IoU metrics. AIS produces credible intervals that consistently achieve high agreement with the MCMC reference, effectively restoring the model uncertainty lost to truncation. Conversely, PB yields only marginal improvements over the RPS for posterior means, but improves competitively on the tail quantiles. We observe similar results in Scenario~1 (see Appendix~\ref{sec:appendix_simulation}).

\subsubsection{Runtime Comparison}
We evaluate the computational cost of each method by comparing the average wall-clock runtime across 50 replications. 
In Scenario~2, while restricted to the RPS only yields poor inference, Rashomon-seeded annealing maintains a highly favorable efficiency profile, requiring approximately 180 seconds on average. This represents roughly a 70\% reduction in runtime compared to the MCMC benchmark, while achieving nearly identical posterior accuracy. We find the PAC-Bayesian estimator to be the most computationally demanding, averaging over 800 seconds, exceeding even the MCMC. This significant overhead stems from the greedy search algorithm inherent in identifying the localized concentration set, which necessitates repeated evaluations of the neighboring model cadidates. Ultimately, Rashomon-seeded annealing provides an effective balance between accuracy and efficiency, as it estimates the full posterior characteristics lost to truncation while preserving a substantial speed advantage over traditional sampling-based methods. All computations were performed on Linux x86\_64 systems (kernel version 4.18.0) equipped with AMD EPYC Genoa processors and approximately 1.5 TB RAM per node.

\section{Application to real data}\label{sec:data_analysis}
We demonstrate the practical utility of Rashomon-seeded annealing using two datasets: (1) telomere length determinants from the National Health and Nutrition Examination Survey (NHANES, \href{https://wwwn.cdc.gov/nchs/nhanes/search/datapage.aspx?Component=Laboratory&Cycle=1999-2000}{Dataset URL}) and (2) a charitable donation dataset \citep{karlan2007does}. 
% In both cases, we show that restricting inference to a truncated Rashomon set can lead to spurious conclusions regarding heterogeneity, whereas our annealed correction provides a more robust and conservative assessment of model uncertainty.
To quantify the stability of our findings across the model space, we utilize the confidence measure proposed in prior RPS literature \citep{apara2025}. This metric calculates the posterior-weighted probability that a specific quantity of interest (such as the difference in telomere length relative to a baseline) falls within specified regions of the cell mean distribution. We partition the range of possible cell means based on the standard deviation of the observed values, including a dedicated region representing a null effect A robust conclusion is signaled by a high confidence mass in a single interval, while ``fragility'', i.e., the sensitivity of results to model choice, is indicated by mass spread across multiple intervals. Besides, the confidence measure, we also compute the $L_1$ error, effective sample size (ESS) ratio, and IoU, which are presented in Table~\ref{tab:single_run}.

\subsection{Analysis of Telomere Length Heterogeneity}\label{sec:nhanes}

We analyze telomere length determinants using NHANES data (1999–2002), focusing on how socioeconomic and demographic factors, such as race, age, gender, education, and work-related stress, influence the T/S ratio, a critical biomarker of cellular aging and immune dysfunction.
Results from the NHANES dataset highlight a critical distinction between truncated and global model exploration. While an initial RPS-restricted analysis suggests several counterintuitive findings, such as Black males possessing longer telomeres than females and a positive correlation between age and telomere length in White populations, analysis using our Rashomon-seeded annealing recalibrates these conclusions (see Figure~\ref{fig:nhanes_difference_heatmap} in the Appendix). By sampling 300 AIS trajectories, we access a broader region of the model space beyond the initial $2,435$ RPS models, revealing that these ``spurious'' effects are likely artifacts of a narrow inferential structure. Notably, the PAC-Bayesian (PB) approach fails to capture these global shifts, yielding results largely identical to the truncated RPS (see Table~\ref{tab:single_run}). Ultimately, by exploring beyond the Rashomon set, our algorithm provides a more robust discovery of heterogeneity, ensuring reported effects are not merely symptoms of model-space truncation.

\subsection{Analysis of the Charitable Donation dataset}
We further evaluate our framework using the Charitable Donation dataset to analyze treatment effect heterogeneity in targeted fundraising. The outcome, individual donation responses, is modeled as a function of demographic and behavioral covariates, including prior engagement and donation history. These features define discrete targeting policies, which are grouped into profiles for pooled inference (see Appendix~\ref{sec:geometry_app}). As shown in Table~\ref{tab:single_run}, Rashomon-seeded annealing achieves competitive accuracy: an $L_1$ error of 0.01 and an IoU of 0.99 relative to the MCMC reference posterior, significantly outperforming the truncated RPS. While the PAC-Bayesian (PB) approach achieves slightly higher accuracy, its runtime is nearly double that of our AIS framework. Ultimately, the AIS correction provides a highly efficient and accurate recovery of posterior uncertainty, essential for robust policy decisions in targeted outreach.

\begin{table}[t]
\centering
\caption{Performance across real-world datasets: $L_1$ error measures discrepancy in posterior mean and IoU measures agreement of credible intervals to that of the MCMC reference posterior (lower $L_1$ and higher IoU are better). The results highlight the trade-off between explored states ($n_{\text{states}}$) and computational efficiency. %Bold values indicate the best non-reference performance.
}
\label{tab:single_run}
\begin{tabular}{lcccccr}
\toprule
{Dataset} & {Method} & $n_{\text{states}}$ & {ESS Ratio} & $L1$ {Error} & {IoU} & {Runtime (s)} \\
\midrule
\multirow{4}{*}{%
\begin{tabular}{@{}c@{}} Charitable \\ Donation \end{tabular}%
}
& RPS & 123  & 1.000 & 0.050 & 0.997 & 6.8 \\
& PB  & 423  & 1.000 & 0.007 & 1.000 & 1373.8 \\
& AIS & 300  & 0.324 & 0.010 & 0.991 & 1199.2 \\
& MCMC & 3000 & -- & -- & -- & 2911.2 \\
\midrule
\multirow{4}{*}{%
\begin{tabular}{@{}c@{}} NHANES \\ Telomere \end{tabular}%
}
& RPS & 2435 & 1.000 & 0.734 & 0.901 & 4.0 \\
& PB  & 2735 & 0.890 & 0.734 & 0.901 & 7172.9 \\
& AIS & 300  & 0.129 & 0.000 & 0.965 & 2914.8 \\
& MCMC & 3000 & -- & -- & -- & 4885.0 \\
\bottomrule
\end{tabular}
\end{table}

\section{Discussion}\label{sec:discussion}
Rashomon‑seeded annealing provides a computationally efficient bridge between deterministic exploration and unguided stochastic sampling. By initializing the annealing process with the Rashomon set, the framework estimates the full posterior distribution with high fidelity while avoiding both the truncation errors inherent in subset‑restricted inference and the prohibitive cost of exhaustive enumeration of the whole model space. The methodology is fundamentally general, and we envision that many other combinatorial model spaces with similar structural topologies, such as trees or directed acyclic graphs, can be accommodated. Implementation in these domains requires a well‑defined distance metric that suitably computes dissimilarities between models and the availability of proposal distributions that permit efficient sampling at each temperature of the annealing trajectory.
While Rashomon sets have traditionally served towards interpretability and predictive robustness, our work highlights a distinct and complementary role, that they can serve as a computational catalyst for estimating the posterior distribution over the whole model space. While our focus here is limited to Bayesian model averaging, investigating how this approach might work for discrete optimization or structure learning is a natural next step for future research. 

\section*{Acknowledgement}
This work used computational resources and storage services of the Hyak Klone cluster provided by the University of Washington and the eScience Institute. The authors were supported by grants from the Office of Naval Research, the U.S. Department of Energy, and the National Institutes of Health. Views expressed in the paper are solely those of the authors.

%%%%%%%%%%%%% BIBLIOGRAPHY %%%%%%%%%%%%%%%%

\bibliographystyle{plainnat}
\bibliography{refs}

%%%%%%%%%%%%% APPENDIX %%%%%%%%%%%%%%%%

\newpage
\appendix

\renewcommand{\thesection}{A\arabic{section}}
\renewcommand{\theequation}{A\arabic{equation}}
\setcounter{equation}{0}
\renewcommand{\thelemma}{A\arabic{lemma}}
\setcounter{lemma}{0}
\renewcommand{\thecorollary}{A\arabic{corollary}}
\setcounter{corollary}{0}
\renewcommand{\theproposition}{A\arabic{proposition}}
\setcounter{proposition}{0}
\renewcommand{\theexample}{A\arabic{example}}
\setcounter{example}{0}
\renewcommand{\theremark}{A\arabic{remark}}
\setcounter{remark}{0}
\renewcommand{\thefigure}{A\arabic{figure}}
\setcounter{figure}{0}

\section{Proof of the theoretical results}\label{sec:consistency}

We prove the almost‑sure consistency of the self‑normalized AIS estimator
stated in Theorem~\ref{thm:general} and its corollaries.  
The notation is exactly that of Section~\ref{sec:annealing}: $Q$ is the joint
distribution of a single model–weight pair $(M,w)$ produced by
Algorithm~\ref{algo:seeded_ais}, and $\mathbb{Q}$ is the product measure on
$(\mathcal{M}\times(0,\infty))^{\mathbb{N}}$ constructed from independent draws
from $Q$.  All expectations $\E_Q$ are taken with respect to $Q$, and
almost‑sure statements refer to $\mathbb{Q}$.

\subsection{Unbiasedness of the AIS weights}
A standard result for annealed importance sampling \citep{neal2001_ais}
guarantees that there exists a constant
\[
C \;=\; \frac{C_T}{C_0} \;>\; 0,
\]
where $C_T = \sum_{M \in \calM} \tilde{p}(M \given \calD)$ and $C_0 = \sum_{M \in \calM} q_0(M; \calS)$ are the normalizing constants of the unnormalized posterior $\tilde{p}(\cdot \given \calD)$ and the unnormalized initial density $q_0(\cdot; \calS)$, respectively. For any bounded measurable function $\zeta:\calM \to \Real^p$,
\begin{equation}\label{eq:app_unbiased}
\E_{Q}\bigl[ w\,\zeta(M) \bigr] = C\;\E_{M\mid\mathcal{D}}\bigl[\zeta(M)\bigr] \;.
\end{equation}
Setting $\zeta$ as the unit function yields $\E_{Q}[w]=C$.

\subsection{Proof of Theorem~\ref{thm:general}}
Let $\zeta:\calM \to \Real^p$ be a measurable function satisfying
$\E_{M \given \calD}[\|\zeta(M)\|]<\infty$, and assume the moment conditions of
the theorem, $\E_{Q}[w]<\infty$ and $\E_{Q}[\,w\,\|\zeta(M)\|\,]<\infty$, hold.
Define the vector‑valued sample averages
\[
A_J = \frac{1}{J}\sum_{j=1}^{J} w_j\,\zeta(M_j),\quad
B_J = \frac{1}{J}\sum_{j=1}^{J} w_j \;.
\]
The pairs $(w_j,\zeta(M_j))$ are i.i.d. under $\mathbb{Q}$ because each AIS
trajectory is independently and identically distributed according to $Q$.
Hence, by Kolmogorov's strong law of large numbers applied component‑wise,
\begin{equation}\label{eq:slln_A}
\begin{split}
A_J &\;\xrightarrow{\text{a.s.}}\; \E_{Q}[w\,\zeta(M)],\\
B_J &\;\xrightarrow{\text{a.s.}}\; \E_{Q}[w] \;=\; C.
\end{split}
\end{equation}
The limits are finite by the assumed moment conditions, and $C>0$ by the
definition of the normalizing constants.  
Using the unbiasedness property~\eqref{eq:app_unbiased} with $\zeta$ itself,
\[
\E_{Q}[w\,\zeta(M)] = C\,\E_{M\given\calD}[\zeta(M)] = C\,\psi,
\]
where $\E_{M \given \calD}[\zeta(M)] = \psi$. Since $C>0$, the continuous
mapping theorem applied to $A_J/B_J$ yields
\[
\hat{\psi}_J = \frac{A_J}{B_J} \xrightarrow{\text{a.s.}} \frac{C\,\psi}{C} = \psi,
\]
where the convergence is with respect to the product measure $\mathbb{Q}$ and
holds in $\Real^p$.

\subsection{Proof of Corollary~\ref{cor:posterior_func}}
For a scalar parameter $\theta$, define
$\zeta_\theta(M)=P(\theta\mid\mathcal{D},M)$, the conditional posterior
distribution function of $\theta$ under model $M$.
Because $0\le\zeta_\theta(M)\le 1$, $\zeta_\theta$ is bounded, and hence
\[
\E_{Q}\bigl[\,w\,\|\zeta_\theta(M)\|\,\bigr] \le \E_{Q}[w] = C < \infty \;.
\]
The moment condition on $w\|\zeta\|$ in Theorem~\ref{thm:general} is therefore
automatically satisfied whenever $\E_{Q}[w]<\infty$.  Applying
Theorem~\ref{thm:general} with $p=1$ and $\zeta=\zeta_\theta$ immediately gives
\[
\hat{P}_J(\theta) =
\frac{\sum_{j=1}^J w_j\,P(\theta\given\calD,M_j)}
     {\sum_{j=1}^J w_j} \xrightarrow{\text{a.s.}} \E_{M \given \calD}[P(\theta \given \calD, M)] = P(\theta \given \calD) \;.
\]

\subsection{Proof of Corollary~\ref{cor:uniform}}
The model space $\calM$ is finite. Hence the posterior over models is a
discrete distribution with probabilities $p(M\mid\mathcal{D})$, and the
marginal posterior distribution function of $\theta$ is the finite mixture
\[
P(\theta\mid\mathcal{D}) = \sum_{M\in\mathcal{M}} P(\theta\mid\mathcal{D},M)\,
p(M\mid\mathcal{D}).
\]

We note that, the AIS estimator for $P(\theta\mid\mathcal{D})$ can be rewritten as
\[
\hat{P}_J(\theta) = \sum_{M\in\calM} P(\theta\given\calD,M)\,
\hat{p}_J(M),
\qquad
\hat{p}_J(M)=\frac{\sum_{j=1}^J w_j\,\mathbf{1}(M_j=M)}{\sum_{j=1}^J w_j}.
\]
Because $0\le P(\theta\mid\mathcal{D},M)\le 1$ for every $\theta$ and $M$,
\[
\sup_{\theta} \bigl| \hat{P}_J(\theta) - P(\theta\mid\mathcal{D}) \bigr|
\;\le\;
\sum_{M\in\mathcal{M}} \bigl| \hat{p}_J(M) - p(M\mid\mathcal{D}) \bigr|
\;=\;
\| \hat{p}_J - p \|_{\mathrm{TV}},
\]
which follows from the triangle inequality.  Here $\|\cdot\|_{\mathrm{TV}}$
denotes the total variation distance.

Now, for each fixed $M$, Theorem~\ref{thm:general} with
$\zeta(M') = \mathbf{1}(M' = M)$, which is bounded, yields
$\hat{p}_J(M) \xrightarrow{\text{a.s.}} p(M\mid\mathcal{D})$ as $J\to\infty$.
Since $\mathcal{M}$ is finite, the total variation distance is the finite sum
of absolute differences, hence
$\|\widehat p_J-p\|_{\mathrm{TV}}\xrightarrow{\text{a.s.}}0$,
and therefore
\[
\sup_{\theta} \bigl| \hat{P}_J(\theta) - P(\theta\mid\mathcal{D}) \bigr|
\;\xrightarrow{\text{a.s.}}\; 0.
\]

The quantile consistency follows from uniform convergence together with
continuity of the limit.  Because $P(\theta\mid\mathcal{D})$ is continuous and
strictly increasing in a neighbourhood of $q_\alpha$, for any $\varepsilon>0$
there exists $\delta>0$ such that
$P(q_\alpha - \varepsilon\mid\mathcal{D}) \le \alpha - \delta$ and
$P(q_\alpha + \varepsilon\mid\mathcal{D}) \ge \alpha + \delta$.
The uniform convergence
$\sup_\theta|\hat{P}_J(\theta) - P(\theta\mid\mathcal{D})|
\xrightarrow{\text{a.s.}}0$ then implies that, for sufficiently large $J$,
$\hat{P}_J(q_\alpha - \varepsilon) < \alpha$ and
$\hat{P}_J(q_\alpha + \varepsilon) \ge \alpha$ almost surely, which forces
$\hat{q}_{\alpha,J} \in (q_\alpha-\varepsilon, q_\alpha+\varepsilon)$.
Hence $\hat{q}_{\alpha,J} \xrightarrow{\text{a.s.}} q_\alpha$.

\subsection{Discussion of the assumptions}
The consistency result relies on the existence of the first moment of $w$ and
of $w\,\|\zeta(M)\|$ under the sampling distribution $Q$.  In practice, the
temperature ladder and the number of Metropolis–Hastings steps $L$ are chosen
so that the importance weights have moderate variability; a small variance of
$\log w$ is a standard diagnostic.
If $\E_{Q}[w]<\infty$, the strong law of large numbers applies without further
moment restrictions, and the existence of $\E_{Q}[w\|\zeta\|]$ is guaranteed
whenever $\zeta$ is bounded, a condition met by all distribution‑function
targets as well as by many other functionals of interest (moments, probabilities
of specific model configurations, etc.).  For unbounded $\zeta$, a finite
moment of $w\|\zeta\|$ can be verified empirically or ensured by truncation
arguments.  The product measure $\mathbb{Q}$ formalises the notion of repeating
the AIS procedure infinitely often, and the almost‑sure convergence means that
for almost every realisation of the infinite sequence of models and weights,
$\hat{\psi}_J$ approaches $\psi$.
In finite samples, the quality of approximation is monitored through the
effective sample size $(\sum_{j=1}^J w_j)^2 / \sum_{j=1}^J w_j^2$ and by
checking the stability of the estimates across independent runs.
Finally, the constant $C$ cancels in the self‑normalised estimator, so the
user never needs to compute the normalising constants $C_T$ and $C_0$; only the
unnormalised posterior and initial density are required to form the importance
weights.  This feature makes the method fully practical for the large discrete
spaces considered in this paper.

\section{Additional Details of AIS}\label{sec:ais_app}

\subsection{Estimation of quantiles from mixture distribution}
To support quantile estimation from the weighted posterior sample, we consider a mixture of $J$ continuous probability distributions with cumulative distribution functions (CDFs) $F_j(\cdot)$, each depending on $\theta$, and
normalized weights $\bar w_j = w_j / \sum_{k} w_k$.  The mixture CDF at a point
$\theta$ is
\[
  \hat P_J(\theta) = \sum_{j=1}^J \bar w_j \, F_j(\theta) \;.
\]
For a target probability $\alpha \in (0,1)$, we seek the mixture quantile
$q_\alpha$ such that $\hat P_J(q_\alpha) = \alpha$.

\begin{lemma}[Quantile Bracketing]\label{lemma:quantile_bracketing}
  Let $q_{\alpha,j}$ be the $\alpha$-th quantile of the $j$-th component, i.e.,
  $F_j(q_{\alpha,j}) = \alpha$.  Then the mixture quantile $q_\alpha$ satisfies
  \[
    \min_{j} q_{\alpha,j} \le q_\alpha \le \max_{j} q_{\alpha,j} \;.
  \]
\end{lemma}
\begin{proof}
  Each $F_j$ is non‑decreasing.  If $\theta < \min_j q_{\alpha,j}$, then
  $F_j(\theta) \le F_j(q_{\alpha,j}) = \alpha$ for all $j$, with strict
  inequality for at least one $j$ provided the CDFs are strictly increasing near
  the quantile.  Hence $\hat P_J(\theta) < \alpha$.  By symmetry, if
  $\theta > \max_j q_{\alpha,j}$ then $\hat P_J(\theta) > \alpha$.  Because
  $\hat P_J$ is continuous, the intermediate value theorem guarantees that
  $q_\alpha$ lies within the stated bounds.
\end{proof}

In practice, we set the search interval $[a,b]$ as
\[
  a = \min_j q_{\alpha,j} - \upsilon,\quad
  b = \max_j q_{\alpha,j} + \upsilon \;,
\]
with a small numerical margin $\upsilon$.  This ensures that the objective
function $f(\theta) = \hat P_J(\theta) - \alpha$ changes sign on $[a,b]$,
permitting robust root‑finding.  When the component CDFs belong to a
location‑scale family (e.g., Gaussian in our factorial design
example), we compute component quantiles via standard inverse‑CDF
routines and solve $f(\theta)=0$ using Brent's method. We implement this routine by utilizing the \texttt{ndtr} and \texttt{ndtri} functions from the \texttt{scipy.special} library for efficient evaluation of the Gaussian distribution function, alongside the \texttt{brentq} root-finder from the \texttt{scipy.optimize} package to achieve robust numerical convergence.

\subsection{Adaptive selection of the temperature ladder}\label{sec:adaptive_ladder}

The efficiency of AIS hinges on the temperature schedule $\{b_t\}_{t=1}^T$.
Excessively large steps produce weight degeneracy, while overly small steps
waste computation.  We adapt the ladder by monitoring the effective sample size
(ESS) of the importance weights that a candidate step would generate, using a
small set of pilot particles \citep{cameron_2019}.

Let $g(M) = \log \tilde{p}(M \given \calD) - \log p_0(M; \calS)$,
where $\tilde{p}(M \given \calD)$ is the unnormalized posterior and
$p_0(M; \calS)$ the seed‑based initial distribution. For a proposed temperature increment $\Delta b$ from $b_{t-1}$ to $b_t = b_{t-1} + \Delta b$,
the change in log‑weight for a model $M$ distributed approximately according to
$p_{t-1}$ is $\Delta b \, g(M)$.  Given a pilot sample of $K$ models
$\{M^{(k)}\}_{k=1}^K$ that are representative of $p_{t-1}$ (initially drawn
from $p_0$, and after each accepted rung refreshed by a few Metropolis–Hastings
steps targeting the new $p_t$), we compute the normalized incremental weights
\[
  \bar w_k(\Delta b) =
  \frac{\exp\{\Delta b \, g(M^{(k)})\}}
       {\sum_{\ell=1}^K \exp\{\Delta b \, g(M^{(\ell)})\}},
  \quad k = 1,\dots,K.
\]
The step‑wise ESS is $\mathrm{ESS}(\Delta b) = 1 / \sum_{k=1}^K \bar w_k(\Delta b)^2$,
and we monitor the ratio $\mathrm{ESS}(\Delta b)/K$.

Starting from $b_0 = 0$, we propose an initial $\Delta b$ (e.g., $0.5$).  If
$\mathrm{ESS}(\Delta b)/K$ falls below a pre‑specified threshold $\tau$
(typically $0.7$–$0.9$), we halve $\Delta b$ and re‑evaluate until the
criterion is met.  The first accepted rung is $b_1 = \Delta b$.  We then
optionally resample the pilot particles according to $\bar w_k(\Delta b)$ and
apply a small number of Metropolis–Hastings moves targeting $p_1$ to keep them
representative of the new intermediate distribution.  The process is repeated
from $b_1$ to determine the next increment, and continues until $b_T = 1$ is
reached.  The resulting ladder is dense where $\log \tilde p - \log p_0$ varies
sharply (preventing weight degeneracy) and sparse where the two distributions
are already close (saving computation).  After the ladder is fixed, the full AIS run with $J$ independent chains is performed using this schedule, yielding the
weighted sample $\{(M_j,w_j)\}_{j=1}^J$ for final posterior estimation.

\begin{remark}
  The AIS output is a properly weighted sample from the full posterior
  $p(M \given \mathcal{D})$; no model space truncation is imposed.  Any posterior
  summary—moments, quantiles, or marginal densities—can be estimated without
  restricting inference to the seed set.  The theoretical guarantees of
  Section~\ref{sec:annealing} hold for any temperature ladder that preserves
  the unbiasedness of the weights; the adaptive procedure merely improves
  finite‑sample efficiency.  Residual Monte Carlo error can be quantified
  through standard diagnostics such as the effective sample size, as illustrated
  in the simulations (Section~\ref{sec:simulation}).
\end{remark}

\section{Geometric foundations of interpretable factorial designs}\label{sec:geometry_app}

The central task for estimating treatment effect heterogeneity in factorial designs as given by \eqref{eq:hier}, is to partition the feature space by pooling cell means $\beta_k$ that are homogeneous, i.e., sharing an identical expected outcome. To maintain scientific coherence, we do not consider arbitrary groupings of cells. Instead, we aggregate only contiguous feature combinations where pooling signifies that the outcome is invariant to incremental shifts in feature levels.

We formalize the structure of our design by focusing on features where the levels possess a natural progression. For such a feature $m$,  the levels are \emph{ordered}, i.e., $0 \prec 1 \prec \dots \prec R-1$. This linear order represents the intuitive progression of ``higher'' intensity for a single intervention. We then equip the entire feature space $\mathcal{K}$ with a \emph{partial order} derived from these level-wise rankings. To be specific, for any two feature combinations $k, k' \in \calK$, we say $k \preceq k'$ if and only if either $k_m \prec k'_m$ or $k_m = k'_{m}$ for every feature $m \in \{1, \dots, M\}$. For features that lack a natural ordering, such as categorical demographic or geographic variables (e.g., country, race), we do not impose a level-wise ranking. Instead, these variables are typically used to define \emph{profiles} (detailed in Section~\ref{sec:profiles}), effectively creating distinct regimes for different levels of the categorical factor. This ensures that our ordering only governs comparisons where ``incremental change'' is well-defined.
\begin{example}\label{ex:1}
Consider a setting where a policymaker evaluates a policy bundle consisting of three distinct interventions, each taking a level of 0 or 1 with $0 \prec 1$. Here, each specific policy represents a feature combination, where the individual interventions correspond to the features. Then, the resulting ordering of the feature space is ``partial'' because it distinguishes between comparable increments and incomparable trade-offs. For example, if the third intervention is fixed at level $x$, the transition from policy $(0, 0, x)$ to $(1, 0, x)$ represents a clear ``increment'' in the first intervention. Because the levels of all other features are held constant or increased, we can say $(0, 0, x) \preceq (1, 0, x)$. On the other hand, the policies $(1, 0, x)$ and $(0, 1, x)$ are incomparable, since one intervention increases while the other decreases, neither policy is strictly ``greater'' than the other.
\end{example}
This structure is well represented by a Hasse diagram \citep{banerjeeTVA}, where the structural skeleton, comprising the nodes and directed edges, corresponds to feature combinations and their ordering, respectively (see Figure~\ref{fig:hasse_main}). 
This allows us to visually distinguish comparable increments from incomparable pairs and restrict pooling to well-defined paths of incremental change. 

\subsection{Profile}\label{sec:profiles}
To provide a principled framework for constructing the model space, we introduce the concept of a profile. Profiles allow the researcher to encode specific scientifically meaningful regimes -- contexts in which different subsets of features are expected to operate.
\begin{definition}[Profile]\label{def:profile}
A \emph{profile}, $\rho \in \{0, 1\}^M$, is a binary vector indicating, for each of the $M$ features, whether that feature is ``active'' (1) or ``inactive'' (0) in the Hasse diagram.
\end{definition}
For any feature combination $k \in \mathcal{K}$, let $\profile(k)$ denote the mapping that returns the profile associated with $k$. Given a target profile $\rho$, we define the subset of the feature space $\mathcal{K}_\rho \subseteq \mathcal{K}$ consisting of all feature combinations $k$ that are consistent with the fixed levels of the inactive features in $\rho$, i.e., $\calK_\rho = \{k \in \calK : \profile(k) = \rho\}$. Here, $\mathcal{K}_\rho$ essentially corresponds to a conditional Hasse diagram, i.e., a specific ``slice'' of the full factorial design. For example, each Hasse diagram in Figure~\ref{fig:hasse_main} corresponds to the profile $\rho = (1, 1, 0)$ since the third feature is fixed at $x$ and only the first and second features change values. By defining profiles a priori, we anchor the analysis to the researcher's design-based representation. This prevents the accidental aggregation of non-comparable treatment configurations and ensures that the resulting partitions reflect substantive scientific contrasts rather than statistical coincidences.

\subsection{Pools and partitions}
While the Hasse diagram provides the structural scaffold of the design, the statistical model is defined by how we group feature combinations into regions of homogeneous outcome. Consider a scenario in setting of Example~1, where transitioning the first intervention from level 0 to 1 while the third intervention remains fixed at $x$, moving from $(0, 0, x)$ to $(1, 0, x)$, yields a negligible shift in the expected outcome. In such instances, the two policies are effectively functionally equivalent, and we can collapse them into a single unit of inference. By merging these redundant feature combinations, we condense the design space and eliminate granularity that do not contribute to the outcome.
Hence, given a profile $\rho$, we define the following.
\begin{definition}[Pool]
    A pool $\pi$ is a set of feature combinations having identical expected outcomes, i.e., for any $k^{(1)}, k^{(2)} \in \pi$, $\beta_{k^{(1)}} = \beta_{k^{(2)}}$.
\end{definition}
\begin{definition}[Partition]\label{def:partition}
A partition $\Pi = \{\pi_1, \pi_2, \dots, \pi_H\}$ is a collection of disjoint pools that together span the profile, such that $\bigcup_{h=1}^H \pi_h = \calK_\rho$.
\end{definition}
To be explicit, the hierarchy between profiles, partitions, and pools is as follows: a profile fixes which features are ``active'' or are being considered; within a profile, a partition divides the active feature combinations, and; each element of a partition is a pool with constant expected outcome.
We illustrate pooling in Figure~\ref{fig:hasse_main} by assigning identical colors to feature combinations within the same pool. In Figure~\ref{fig:hasse1}, for instance, the aggregation of $(0, 0, x)$ with $(1, 0, x)$ indicates that these policies share a common cell mean; similarly, $(0, 1, x)$ and $(1, 1, x)$ are pooled together. On the other hand, Figure~\ref{fig:hasse3} represents the ``saturated'' case under the profile corresponding to the third feature being fixed at $x$; in this instance, every feature combination constitutes its own pool, signifying that each distinct configuration yields unique expected outcome.

While identifying optimal partitions is our focus, discovering profiles, i.e., determining active versus fixed features, is an NP-hard problem beyond our scope. Without specific scientific hypotheses, this entails learning a unique Hasse diagram for every combination of feature levels. We therefore restrict our analysis to partitions within a single profile $\calK_\rho$. Although one can pool across profiles \citep{apara2025}, we focus on partitions reflecting the researcher's \emph{a priori} design choices.

\subsection{Permissible partitions}
Not all partitions are scientifically meaningful. While any grouping of feature combinations is technically a partition, we limit our model space to those that are interpretable and substantively coherent.
We define permissibility through the geometry of the Hasse diagram. The key idea is that we can define these restrictions in terms of the edges denoting the paths of incremental change, rather than treating nodes as isolated points. Let $\calP^\ast$ denote the set of all permissible partitions.
\begin{definition}[Permissible Partition]\label{def:permissibility}
A partition $\Pi$ is \emph{permissible} if and only if it satisfies the following conditions:
\begin{enumerate}
  \item \textbf{Homogeneity:} Every $\pi \in \Pi$ is a pool where all $k \in \pi$ share an identical expected outcome.
  \item \textbf{Contiguity:} Every $\pi$ is a \emph{closed interval} in the partial order. This means each pool forms a contiguous block with no ``holes'' along any monotone path.
  \item \textbf{Parallel Splits:} The partition respects parallel symmetry. Formally, for any distinct pools $\pi_i, \pi_j$, if their minima (or maxima) are incomparable, there must exist another pool whose corner sits at their coordinate-wise maximum (or minimum).
\end{enumerate}
\end{definition}
Operationally, a partition is permissible if and only if it can be obtained by cutting a set of parallel edge-families that run all the way through the Hasse diagram. 
\begin{example}
    Figure~\ref{fig:hasse1} satisfies all permissibility conditions. It represents a clean parallel split where the pools are defined by the first intervention. The minima of the pools are comparable, so Condition~3 is trivially satisfied. Figure~\ref{fig:hasse3} is the saturated case; since every pool is a singleton, it trivially satisfies Condition~2 (convexity) and Condition~3 (parallelism), as every incomparable pair of nodes has a corresponding singleton pool at their coordinate-wise maximum. On the other hand, the configuration in Figure~\ref{fig:hasse2} violates Condition~3 (parallelism). Here, the pool containing $(0, 1, x)$ and $(1, 1, x)$ and the singleton pool $(1, 0, x)$ have incomparable minima, thus violating Condition~3. In this case, we pool $(0, 1, x)$ and $(1, 1, x)$, implying the first intervention has zero effect when the second feature is at level 1. However, we do not pool $(0, 0, x)$ and $(1, 0, x)$, implying the same intervention does have an effect when the second feature is at level 0. This partial split implies a measure-zero interaction where treatment effects perfectly offset only at specific levels, which we exclude for lack of robustness.
\end{example}

\section{Derivations of posterior distributions}\label{sec:appendix_posterior}

Given the Bayesian hierarchical model given in \eqref{eq:hier}, we derive the posterior distributions for the pool means $\gamma$ and the cell means $\beta$ given a fixed value $\hat{\sigma}^2$ for $\sigma^2$. The derivation follows from the conjugacy of the $g$-prior. For a fixed partition $\Pi = \{\pi_1, \ldots, \pi_H\} \in \calP^\ast$, the linear model is
\begin{equation}\label{eq:model_app}
  y = \tilde{X} \beta + \epsilon, \quad \epsilon \sim \Normal(0, \sigma^2 I_n),
\end{equation}
where $\tilde{X} \in \Real^{n\times K}$ is the design matrix, $\beta \in \Real^K$ is the vector of cell means, and the constraint that all cells in the same pool share the same mean is expressed as $\beta = \Lambda_\Pi \gamma$, with $\gamma \in \Real^H$ the vector of pool means and $\Lambda_\Pi \in \{0,1\}^{K\times H}$ the binary partition matrix, as defined in Section~\ref{sec:bayes_model}.
Define the reduced design matrix $\tilde{X}_\Pi \coloneqq \tilde{X}\Lambda_\Pi$ so that $D\beta = \tilde{X}_\Pi\gamma$.  We assume $n > H$. 

\begin{lemma}\label{lemma1}
    Each row of $\tilde{X}_\Pi$ contains exactly one 1.
\end{lemma}
\begin{proof}
For row $i$, let $k(i)$ be the unique column with $\tilde{X}_{i,k(i)} = 1$. Then
\[
(\tilde{X}_\Pi)_{ih} = \sum_{k=1}^K \tilde{X}_{ik} (\Lambda_\Pi)_{kh} = (\Lambda_\Pi)_{k(i),h}.
\]
Now, row $k(i)$ of $\Lambda_\Pi$ has exactly one $1$, since there exists only one pool, say $\pi_h$, such that $k(i) \in \pi_h$. Thus, we have $(\Lambda_\Pi)_{k(i), h}=1$, and all other entries are $0$.  Hence $(\tilde{X}_\Pi)_{ih}=1$ for that $h$ and $0$ otherwise.
\end{proof}

\begin{proposition}\label{prop1}
    $\tilde{X}_\Pi^{\T} \tilde{X}_\Pi = \operatorname{diag}(n_1,\ldots,n_H)$, where $n_h = \bigl|\{i : k_i \in \pi_h\}\bigr|$.
\end{proposition}
\begin{proof}
Expanding the $h,j$-th entry of $\tilde{X}_\Pi^\T \tilde{X}_\Pi$, we have
  \[
    (\tilde{X}_\Pi^{\T} \tilde{X}_\Pi)_{hj} = \sum_{i=1}^n (\tilde{X}_\Pi)_{ih}\,(\tilde{X}_\Pi)_{ij}.
  \]
  By Lemma~\ref{lemma1}, for a given $i$ the product $(\tilde{X}_\Pi)_{ih}(\tilde{X}_\Pi)_{ij}$ is non‑zero only if $h=j$ and observation $i$ belongs to pool $h$.  Consequently, we have
  \[
    (\tilde{X}_\Pi^{\T} \tilde{X}_\Pi)_{hj} = 0 \quad \text{if} \ h \neq j,
    \quad
    (\tilde{X}_\Pi^{\T} \tilde{X}_\Pi)_{hh} = \sum_{i=1}^n (\tilde{X}_\Pi)_{ih} = n_h,
  \]
  which is precisely the number of observations whose feature combination lies in pool $\pi_h$.
\end{proof}

\begin{proposition}\label{prop2}
  $\tilde{X}_\Pi^{\T} y = (S_1,\dots,S_H)^{\T}$, where $S_h = \sum_{i:k(i)\in\pi_h} y_i$.
\end{proposition}
\begin{proof}
By Lemma~\ref{lemma1}, if $k(i) \in \pi_h$, then 
  \[
    (\tilde{X}_\Pi^{\T} y)_h = \sum_{i=1}^n (\tilde{X}_\Pi)_{ih}\, y_i = \sum_{i: k(i)\in\pi_h} y_i = S_h.
  \]
\end{proof}

\subsection{Conditional posterior distribution of pool means}\label{sec:gamma_app} 

Treating $\sigma^2$ as known, from the hierarchical model~\eqref{eq:hier}, we have the likelihood and the prior as
\begin{equation*}
\begin{split}
    y \given \gamma, \Pi, \sigma^2 &\sim \Normal \left( \tilde{X}_\Pi\gamma, \sigma^2 I_n \right),\\
    \gamma \given \Pi, \sigma^2 &\sim \Normal \left( 0, g\sigma^2\,(\tilde{X}_\Pi^\top \tilde{X}_\Pi)^{-1} \right).
\end{split}
\end{equation*}
Following the conjugacy of the Gaussian prior on a Gaussian likelihood, we have the posterior precision of $\gamma$ as
\[
\Sigma_\gamma^{-1} = \frac{1}{\sigma^2} \tilde{X}_\Pi^\T \tilde{X}_\Pi + \frac{1}{g\sigma^2}\tilde{X}_\Pi^\T \tilde{X}_\Pi
   = \frac{g+1}{g\sigma^2} \tilde{X}_\Pi^\T \tilde{X}_\Pi \;.
\]
Moreover, the posterior mean of $\gamma$ is given by the weighted mean
\[
\hat{\gamma}_{\Pi} = \Sigma_\gamma \left( \frac{1}{\sigma^2} \tilde{X}_\Pi^\T y \right) = \frac{g}{g+1} \hat\gamma \;.
\]
Therefore, we have the conditional posterior distribution of $\gamma$ as
\begin{equation*}
    \gamma \given \Pi, \calD \sim \Normal \left( \frac{g}{g+1} \hat\gamma_{\Pi}, \frac{g}{g+1} \sigma^2 \mathrm{diag} \left(n_1, \ldots, n_{|\Pi|} \right)^{-1} \right) \;,
\end{equation*}
which follows from Proposition~\ref{prop1} above. Equation~\ref{eq:post_beta} follows from the transformation $\beta = \Lambda_\Pi \gamma$.

\subsection{Marginal posterior of partitions}
\label{sec:equivalence_app}

Treating $\sigma^2$ as known, the marginal likelihood $p(\calD \given \Pi, \sigma^2) = \int p(\calD \given \gamma, \Pi) p(\gamma \given \Pi, \sigma^2) \mathrm{d}\gamma$ is obtained by combining the two normal densities.  The exponent of the product is
\[
-\frac{1}{2\sigma^2}\|y-\tilde{X}_\Pi\gamma\|^2 - \frac{1}{2g\sigma^2}\gamma^\top \tilde{X}_\Pi^\top \tilde{X}_\Pi\gamma \;.
\]
Completing the square in $\gamma$ yields
\[
p(\calD \given \Pi, \sigma^2) = (2\pi\sigma^2)^{-n/2} (1+g)^{-|\Pi|/2} \exp \left\{ -\frac{1}{2\sigma^2(1+g)} \left( \|y\|^2 + g\, \mathrm{SSE}_\Pi \right) \right\},
\]
where $\mathrm{SSE}_\Pi = y^\top(I-H_\Pi)y$ with $H_\Pi = \tilde{X}_\Pi(\tilde{X}_\Pi^\top \tilde{X}_\Pi)^{-1}\tilde{X}_\Pi^\top$, i.e., the within‑pool sum of squared errors. Following Propositions~\ref{prop1}~and~\ref{prop2}, this can be further simplified as
\[
\mathrm{SSE}_\Pi = \sum_{h=1}^{H}\sum_{i:k(i)\in\pi_h} (y_i - \hat\gamma_h)^2,\quad
\hat\gamma_h = \frac{1}{n_{h}}\sum_{i:k(i)\in\pi_h} y_i.
\]
Since $\|y\|^2$ does not depend on $\Pi$, we may absorb it into the normalizing constant and write
\[
p(\calD \given \Pi) \propto (1+g)^{-|\Pi|/2} \exp \left\{ -\frac{g}{2\sigma^2(1+g)} \mathrm{SSE}_\Pi \right\} \;.
\]
Subsequently, multiplying by $p(\Pi)\propto\exp(-\lambda|\Pi|)$ yields the unnormalised posterior
\[
p( \Pi \given \calD) \propto \exp \left\{ -\frac{g}{2\sigma^2(1+g)} \mathrm{SSE}_\Pi - \left( \lambda + \frac{1}{2} \log(1+g) \right) |\Pi| \right\} \;.
\]
Now, note that the posterior distribution $p(\Pi \given \calD)$ is maximized when $- \log p(\Pi \given \calD)$, i.e., the expression $\eta_1 \mathrm{SSE}_\Pi + \eta_2 |\Pi|$ is minimized, where
\[
\eta_1 = \frac{g}{2\sigma^2(1+g)}>0,
\quad
\eta_2 = \lambda + \frac12\log(1+g).
\]
Hence, we can write the following
\[
\arg \max_{\Pi \in \calP^\ast} p(\Pi \given \calD) = \arg \min_{\Pi \in \calP^\ast} \left( \eta_1 \mathrm{SSE}_{\Pi} + \eta_2 |\Pi| \right) = \arg \max_{\Pi \in \calP^\ast} \left( \frac{1}{n} \mathrm{SSE}_{\Pi} + \tilde{\lambda} |\Pi| \right) \;,
\]
for some $\tilde{\lambda} > 0$, where the last equality follows from the fact that the minimization problem is invariant with respect to a scaling factor.

\section{Additional simulation results}\label{sec:appendix_simulation}

Figure~\ref{fig:accuracy_scenario1} shows how well each method recovers the exact posterior in Scenario 1. Note that, the modest size of the model space allows us to explicitly find all model posteriors and evaluate the exact posterior of the cell means directly. While the truncated RPS baseline shows a clear loss in accuracy, our proposed annealing algorithm successfully bridges this gap, reaching error levels comparable to the MCMC reference. This recovery is further evidenced by the interval overlap results in Figure~\ref{fig:IoU_scenario1}A; Rashomon-seeded annealing maintains high IoU values across various seed set sizes, significantly outperforming both the restricted RPS and the PB estimator for estimation of the posterior mean. Figure~\ref{fig:IoU_scenario1}B investigates computational burden, highlighting the efficiency of our approach. Rashomon-seeded annealing is substantially faster than both the PB method and traditional MCMC, demonstrating that using a Rashomon set as a ``warm start'' allows for fairly accurate estimation of full-posterior recovery in a fraction of the time required by MCMC.

\begin{figure}[htbp]
    \centering
    \includegraphics[width=0.95\linewidth]{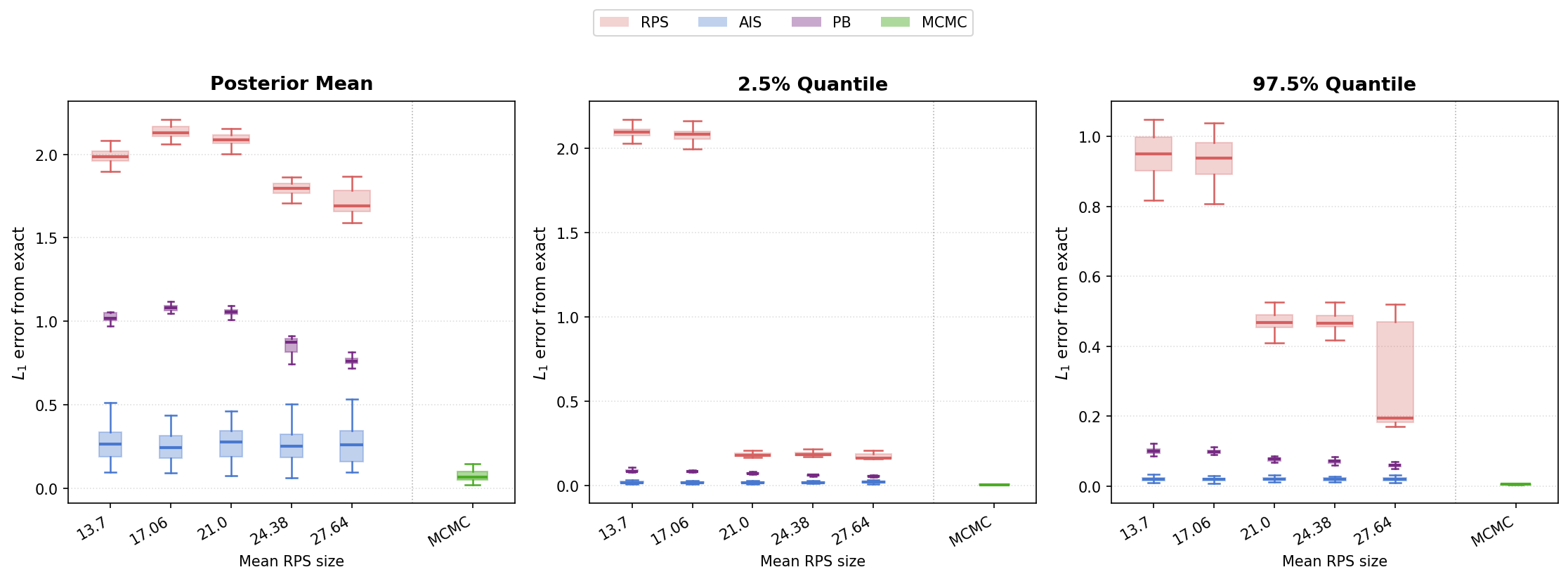}
    \caption{Comparison between the errors from MCMC, RPS, AIS with different $\epsilon$, as compared to the exact posterior}
    \label{fig:accuracy_scenario1}
\end{figure}

\begin{figure}[htbp]
    \centering
    \includegraphics[width=0.75\linewidth]{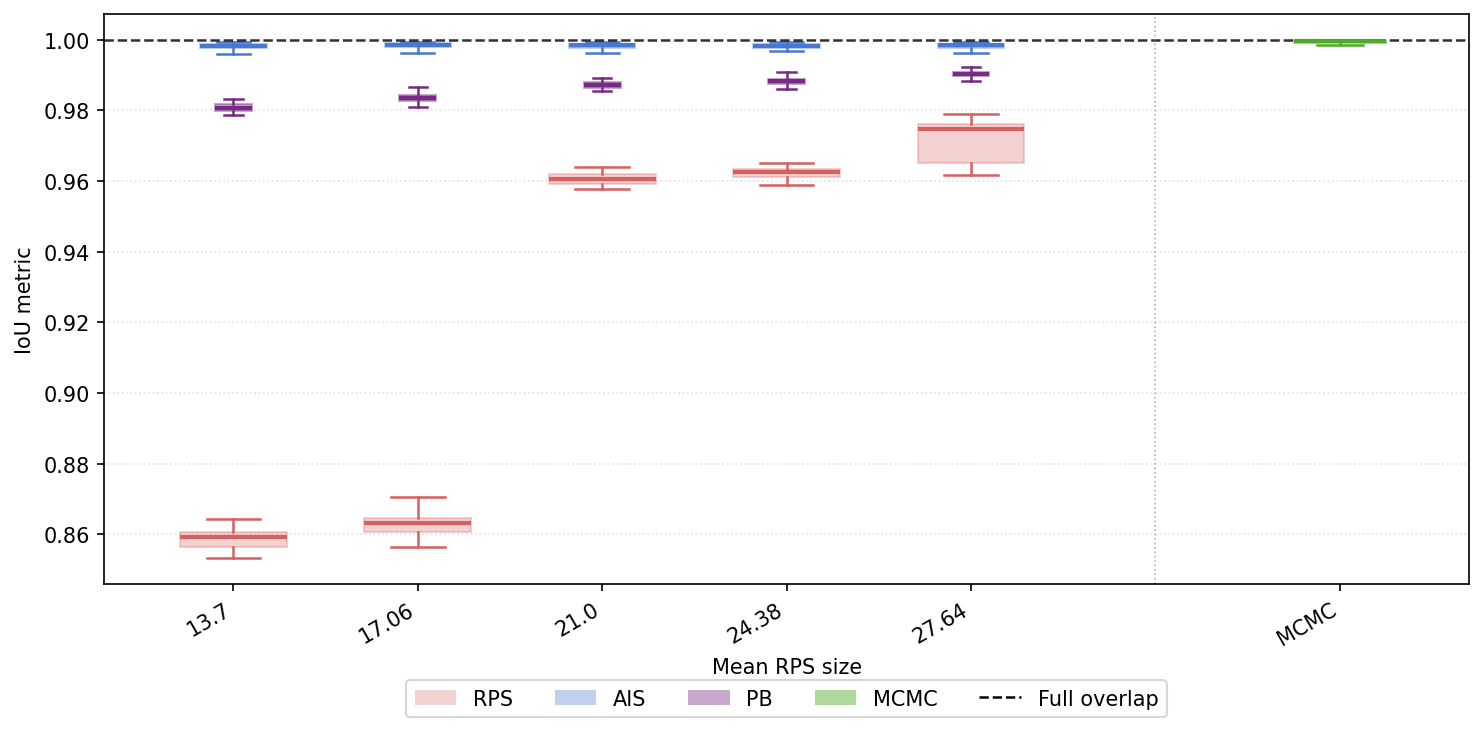}
    \caption{\textbf{A}:Comparison between the IoU metric for overlap with intervals produced by MCMC, for RPS, AIS, PB with different RPS sizes; \textbf{B}:running time for AIS and PB starting with RPS of different sizes, and MCMC}
    \label{fig:IoU_scenario1}
\end{figure}

\section{Additional results from analysis of real data}\label{sec:appendix_real_data}

\begin{figure}[t]
    \centering
    \includegraphics[width=\textwidth]{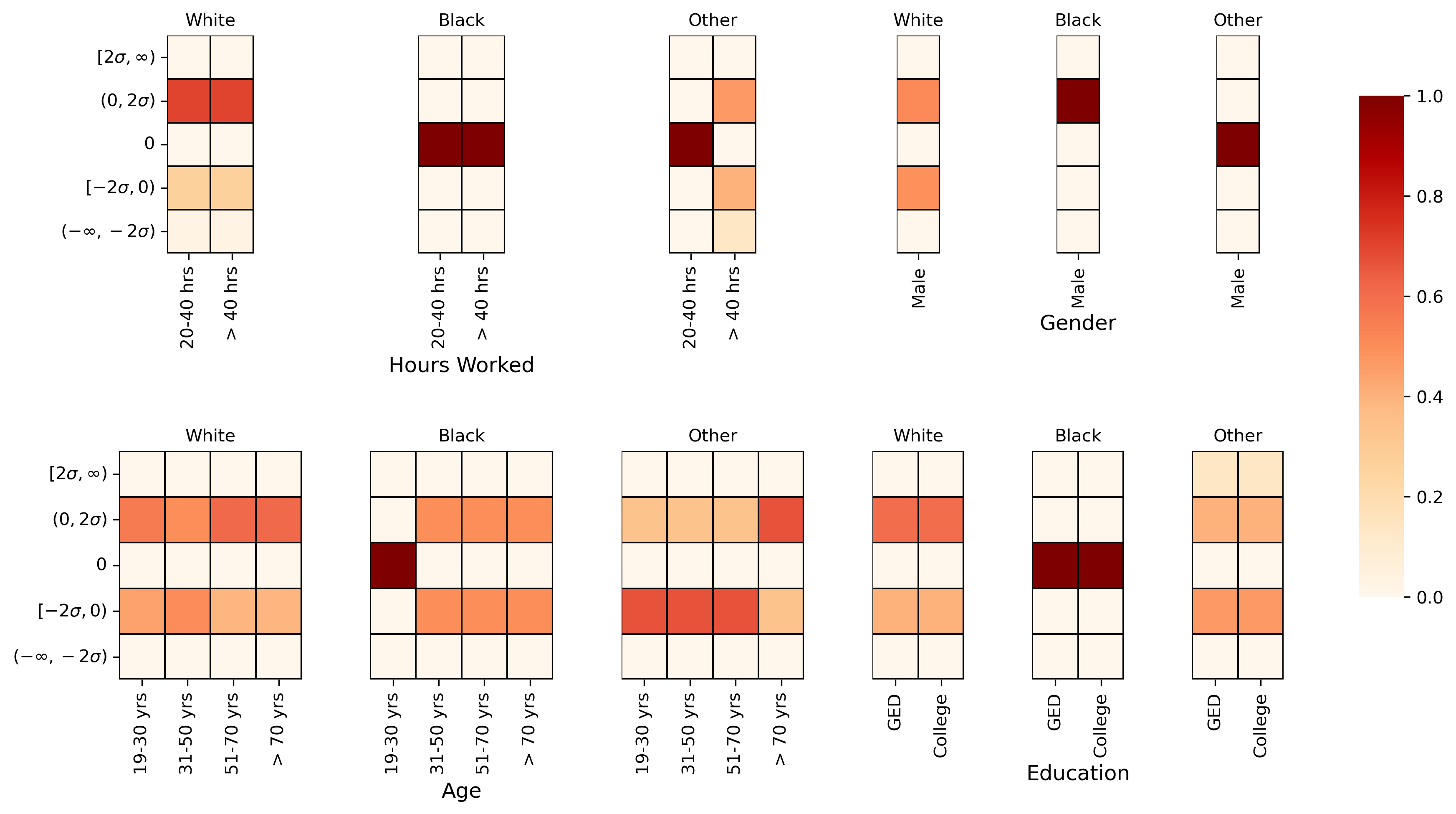}
    \caption{Heterogeneity in telomere length across four features (hours worked,
    gender, age, and education) within the RPS, stratified by race (White, Black,
    Other). Each column corresponds to a feature level, each row
    to one of the five effect-size intervals, and each cell displays
    $c(t_{r,\mathbf{x}}, I)$ for race $r$ and features $\mathbf{x}$. A single
    dark red cell indicates a robust, homogeneous conclusion; multiple lit cells
    indicate heterogeneity or fragility across the RPS.}
    \label{fig:nhanes_heatmap}
\end{figure}

\begin{figure}[t]
    \centering
    \includegraphics[width=\textwidth]{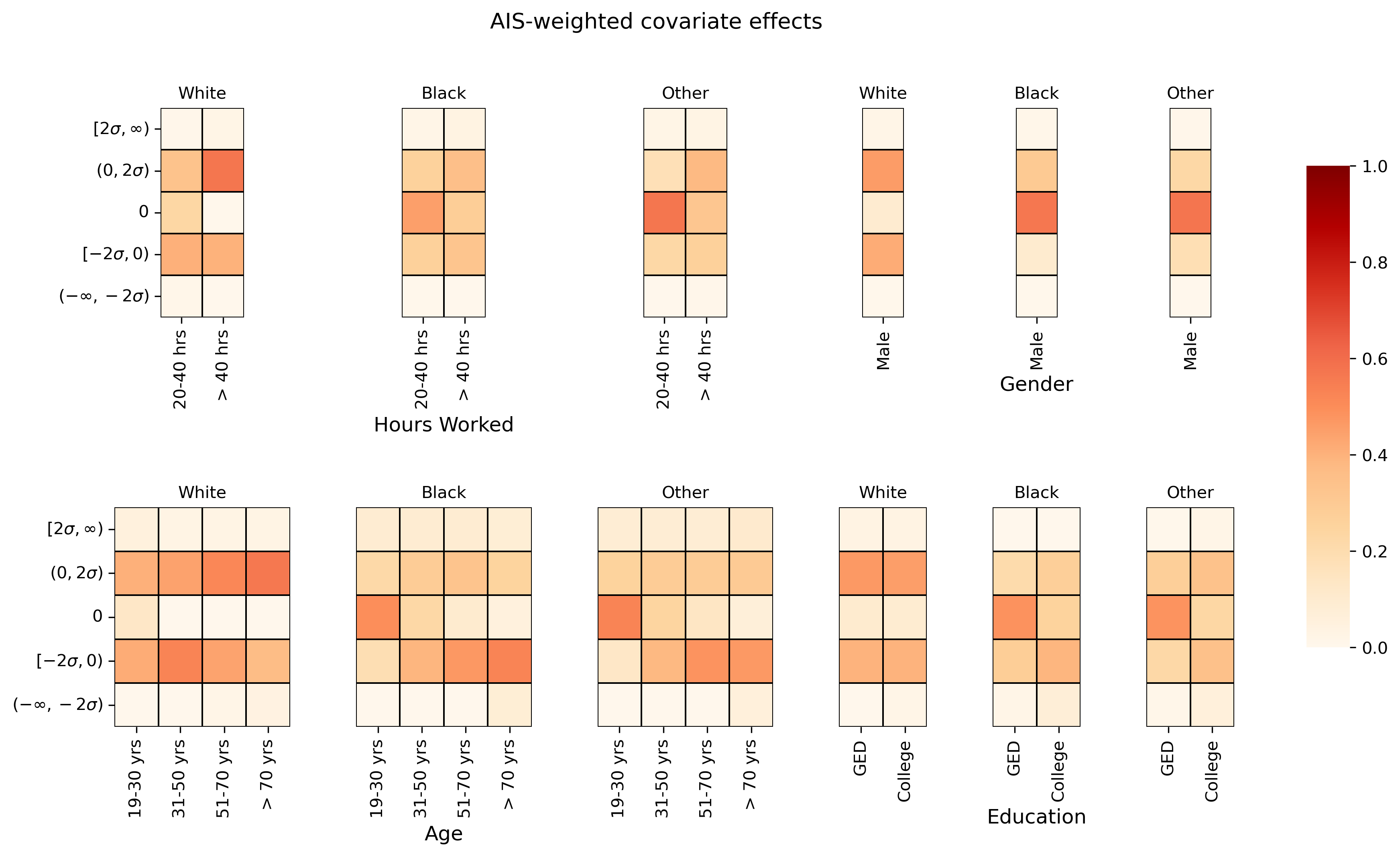}
    \caption{Heterogeneity in telomere length across four features (hours worked,
    gender, age, and education) within the AIS, stratified by race (White, Black,
    Other). Each column corresponds to a feature level, each row
    to one of the five effect-size intervals, and each cell displays
    $c(t_{r,\mathbf{x}}, I)$ for race $r$ and features $\mathbf{x}$. A single
    dark red cell indicates a robust, homogeneous conclusion; multiple lit cells
    indicate heterogeneity or fragility across the RPS.}
    \label{fig:nhanes_ais_heatmap}
\end{figure}

\begin{figure}[t]
    \centering
    \includegraphics[width=\textwidth]{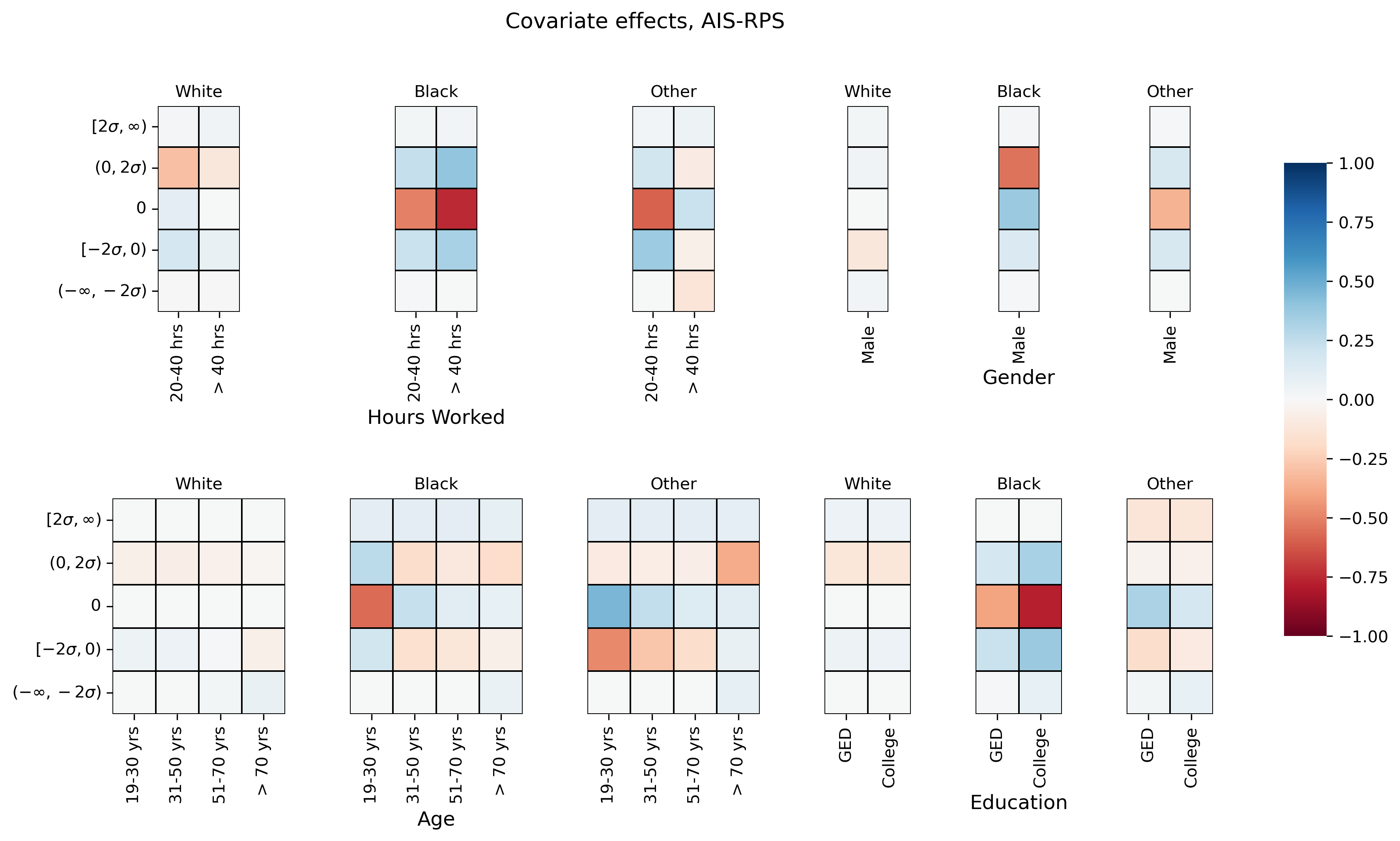}
    \caption{Differences in the estimated variable heterogeneity between AIS and RPS inference}
    \label{fig:nhanes_difference_heatmap}
\end{figure}

The RPS heatmap pools policy-level outcome estimates across all partition matrices within the Rashomon set, weighting each by its model loss -- essentially a loss-weighted average over a discrete, epsilon-bounded collection of models. Because the Rashomon set is defined by a hard threshold, the RPS heatmap is sensitive to the choice of epsilon: models just inside the boundary receive full weight, while equally plausible models just outside are excluded entirely. In contrast, the AIS heatmap is derived from importance-weighted samples drawn from the full g-prior posterior over the model space, with AIS annealing from the RPS-informed prior (which concentrates mass on high-quality partitions) to the posterior target. This means AIS explores states beyond the Rashomon boundary -- including one-step and multi-step neighbors of RPS models -- weighting each terminal state by its true posterior probability under the g-prior, which jointly rewards goodness of fit and parsimony (fewer pools). The result is a heatmap that is less sensitive to the epsilon cutoff, incorporates a principled complexity penalty, and produces posterior credible intervals rather than set-bounded point estimates. Practically, cells where many models in the Rashomon set disagree on the pooling structure show wider AIS credible intervals, while cells with stable pooling across models show tight intervals -- a distinction the RPS heatmap cannot directly express.

\end{document}